\documentclass{IEEEtran}
\usepackage{cite}
\usepackage{amsmath,amssymb,amsfonts,mathrsfs}
\usepackage{algorithm}
\usepackage{algorithmic}
\usepackage{graphicx}
\usepackage{textcomp}
\usepackage{color}
\usepackage{braket}
\usepackage{cuted}%%\stripsep-3pt
\usepackage[colorlinks=true]{hyperref}

\newtheorem{theorem}{Theorem}[section]
\newtheorem{corollary}{Corollary}[section]
\newtheorem{lemma}{Lemma}[section]
\newtheorem{proposition}{Proposition}[section]
\newtheorem{remark}{Remark}[section]
\newtheorem{definition}{Definition}[section]

\newcommand{\bthm}{\begin{theorem}}
\newcommand{\ethm}{\end{theorem}}
\newcommand{\blem}{\begin{lemma}}
\newcommand{\elem}{\end{lemma}}
\newcommand{\bex}{\begin{example}}
\newcommand{\eex}{\end{example}}
\newcommand{\bprop}{\begin{proposition}}
\newcommand{\eprop}{\end{proposition}}
\newcommand{\bplm}{\begin{problem}}
\newcommand{\eplm}{\end{problem}}
\newcommand{\bmrk}{\begin{remark}}
\newcommand{\emrk}{\end{remark}}
\newcommand{\bdfn}{\begin{definition}}
\newcommand{\edfn}{\end{definition}}
\newcommand{\bcor}{\begin{corollary}}
\newcommand{\ecor}{\end{corollary}}

\newcommand{\beq}{\begin{equation}}
\newcommand{\eeq}{\end{equation}}
\newcommand{\beqm}{\begin{equation*}}
\newcommand{\eeqm}{\end{equation*}}
\newcommand{\beqn}{\begin{eqnarray}}
\newcommand{\eeqn}{\end{eqnarray}}
\newcommand{\beqnm}{\begin{eqnarray*}}
\newcommand{\eeqnm}{\end{eqnarray*}}
\newcommand{\bea}{\begin{align}}
\newcommand{\eea}{\end{align}}
\newcommand{\beam}{\begin{align*}}
\newcommand{\eeam}{\end{align*}}
\newcommand{\bs}{\begin{subequations}}
\newcommand{\es}{\end{subequations}}

\newcommand{\bei}{\begin{itemize}}
\newcommand{\eei}{\end{itemize}}
\newcommand{\bed}{\begin{description}}
\newcommand{\eed}{\end{description}}
\newcommand{\bee}{\begin{enumerate}}
\newcommand{\eee}{\end{enumerate}}

\newcommand{\bey}{\begin{array}}
\newcommand{\eey}{\end{array}}

\newcommand{\beb}{\begin{thebibliography}{99}}
\newcommand{\eeb}{\end{thebibliography}}

\newcommand{\mbb}{\mathbb}

\newenvironment{proof}{{\noindent\it Proof}\quad}{\hfill $\square$\par}

\def\BibTeX{{\rm B\kern-.05em{\sc i\kern-.025em b}\kern-.08em
    T\kern-.1667em\lower.7ex\hbox{E}\kern-.125emX}}

\def\ff{\frac}

\newcommand{\la}{\label}

\def\i{{\rm i}}

\graphicspath{{./Figs/}}

\begin{document}

\title{On the Dynamics of the Tavis-Cummings Model}

\author{Zhiyuan Dong, Guofeng Zhang, Ai-Guo Wu, and Re-Bing Wu
%\thanks{This work was supported in part by the National Natural Science Foundation of China under Grants (No. 6217023269, No. 62003111, No. 61833010 and No. 62173201), the Hong Kong Research Grant Council under grants (15208418, and 15203619), and Shenzhen Fundamental Research Fund, China under Grant No. JCYJ20190813165207290. (Corresponding author: Guofeng Zhang.)}
\thanks{Zhiyuan Dong is with School of Science, Harbin Institute of Technology, Shenzhen, Shenzhen, China (e-mail: dongzhiyuan@hit.edu.cn).}
\thanks{Guofeng Zhang is with the Department of Applied Mathematics, The Hong Kong Polytechnic University, Hong Kong. He is also with The Hong Kong Polytechnic University Shenzhen Research Institute, Shenzhen, 518057, China (e-mail: guofeng.zhang@polyu.edu.hk).}
\thanks{Ai-Guo Wu is with School of Mechanical Engineering and Automation, Harbin Institute of Technology, Shenzhen, Shenzhen, China (e-mail: agwu@hit.edu.cn, ag.wu@163.com).}
\thanks{Re-Bing Wu is with the Center for Intelligent and Networked Systems, Department of Automation, Tsinghua University, Beijing, 100084, China (e-mail: rbwu@tsinghua.edu.cn).}
}

\maketitle

\begin{abstract}
The purpose of this paper is to present a comprehensive study of the Tavis-Cummings model from a system-theoretic perspective. A typical form of the Tavis-Cummings model is composed of an ensemble of non-interacting two-level systems (TLSs) that are collectively coupled to a common cavity resonator. The associated quantum linear passive system is proposed,  whose canonical form reveals typical features of the Tavis-Cummings model, including $\sqrt{N}$- scaling, dark states, bright states,  single-excitation superradiant and subradiant states. The passivity of this linear system is related to  the vacuum Rabi mode splitting phenomenon in Tavis-Cummings systems. On the basis of the linear model, an analytic form is presented for the steady-state output  state of the  Tavis-Cummings model driven by a single-photon state.  Master equations are used to study the excitation properties of the  Tavis-Cummings model in the multi-excitation scenario. Finally, in terms of the transition matrix for a linear time-varying system, a computational framework is proposed for  calculating the state of the Tavis-Cummings model, which is applicable to the multi-excitation case.
\end{abstract}

\begin{IEEEkeywords}
Quantum control, Tavis-Cummings model, two-level systems, open quantum systems.
\end{IEEEkeywords}

%
%\bei
%\item \cite{Dicke54}, energy state; coherent radiation accompanies transition between energy levels of the single system.
%\eei

%\tableofcontents

%%%%%%%%%%%%%%%%%%%%%%%%%%
%%%%%%%%%%%%%%%%%%%%%%%%%%
%%%%%%%%%%%%%%%%%%%%%%%%%%
\section{Introduction}\label{sec:intro}

In 1954, Robert Dicke calculated \cite{Dicke54} that an ensemble of gaseous molecules interacting with a common radiation field could exhibit a coherent spontaneous emission process, during which the molecules act as a giant molecule that  shows  superradiation --- cooperative radiation rate much faster than independent individual radiation rates. This problem was further studied by Tavis and Cummings \cite{TC68}  by means of a model of $N$ identical non-interacting two-level systems (TLSs) coupled to a single-mode quantized radiation field, see Fig. \ref{system} in section \ref{sec:system} for an example. An exact solution of the eigenstates of the Hamiltonian for this model is derived. The model proposed in \cite{TC68}  is called the Tavis-Cummings model in the subsequent literature.  The Tavis-Cummings model has been physically realized by quite a few experimental platforms, including superconducting circuits \cite{SPS07, MCG07,Fink09,WW+PRL20},  NV spin ensembles \cite{ANP+17}, and double quantum dots \cite{Deng2015,Van18,WLL+21}. In the Tavis-Cummings model consisting of $N$ TLSs equally coupled to a  cavity resonator, under certain conditions the collective coupling strength of the ensemble exhibits a $\sqrt{N}$-scaling, which can be experimentally observed from the vacuum Rabi mode splitting  of the resonator transmission spectrum of the bright states. In addition to bright states, a Tavis-Cummings model can also have dark states which contain single-excitation subradiant states as a subclass. Applications of the  Tavis-Cummings model can be found in \cite{GH82,MCG07,Van18,WW+PRL20} and references therein.

%A superradiant state is a special linear combination of bright states;  and a  subradiant states a special linear combination of dark states.   ?????

%????? applications of the  Tavis-Cummings model in  The bus structure  of the Tavis-Cummings model is useful for long-distance quantum information processing \cite{MCG07,Van18,CLSW18}.

In this paper, we aim to introduce the Tavis-Cummings model to the quantum control community and show that many of its typical properties can be uncovered by means of systems theory. The main contributions are summarized below.

In section \ref{subsec:linear} we propose a quantum linear system that is associated to the Tavis-Cummings model. This linear model reveals the $\sqrt{N}$-scaling of the coupling strength of the atomic ensemble. Moreover, a transfer function is defined for a performance variable for which the system is passive, and the transfer function reflects the vacuum Rabi mode spitting of the Tavis-Cummings model.   Finally, the structural decomposition of the linear model shows that the bright states of the Tavis-Cummings model live in the controllable and observable subspace, whereas the dark states reside in the uncontrollable and unobservable subspace of the linear model.

In section \ref{sec:output} we apply the quantum linear systems theory to derive an analytic form of the output single-photon state of the Tavis-Cummings system driven by a single-photon input.  The simulations in Fig. \ref{distri1} show that the input photon tends not to interact with the atoms when the number of atoms is large, and photon-atom interaction is easier when atoms are non-resonant. On the other hand, when only one of  the two-level atoms is initially excited and the input field is vacuum, an analytic form of the joint system-field state is given in section \ref{sec:super}, which explains several experimental observations including superradiance and subradiance \cite{WW+PRL20}. %Superradiant and subradiant states can also be obtained by means of linear quantum systems theory developed in  subsection \ref{subsec:linear}; see Remark \ref{rem: sup_sub} for details.

In section \ref{sec:excitation}, we study the excitations of TLSs by means of master equations. When all the $N$ atoms are initially in the excited state, we prove that eventually they all  settle to the ground state and the output field is in an $N$-photon state, provided that the coupling strengths are identical.  %In general, when the coupling strengths are not identical or not all atoms are initially excited, the steady states of the atomic ensemble and the output photons will be entangled. This is interesting as it demonstrates entanglement between standing and flying qubits.

In our preliminary study \cite{LZW20},  a computational framework is proposed to calculate the joint system-field state of a general open quantum system. Here, we develop it further in section \ref{sec:alter} and apply it to the study of the  Tavis-Cummings model. In particular, we derive  the exact form of 2- and 3-photon states. %The problem of continuous-mode two-photon state generation has been studied, both experimentally \cite{FHW+17} and  theoretically \cite{FTR+18}. It is shown that under strong coherent driving pulses, a two-level system coupled to a unidirectional waveguide (namely, single input channel) is able to emit two photons. An explicit form of the two-photon state is given in \cite[(165)]{FTR+18}.  %Scattering of two and three photons from a ladder-type finite-level quantum system has been calculated in \cite{PZ17}, where the analytical forms of  2- and 3-photon states are given; see \cite[(19) and (27)]{PZ17} for details. In general cases, it is difficult to derive explicit expressions of continuous-mode 2- and 3-photon states.

{\it Notation.} The reduced Planck constant $\hbar$ is set to $1$. $\mathrm{i}=\sqrt{-1}$ is the imaginary unit, $\delta_{ij}$ denotes the Kronecker delta function, and $\delta(t-r)$ is the Dirac delta function. $I$ is the identity matrix, and $\mathbf{0}$ is the zero vector or matrix whose dimension can be easily determined from the context.  Given a column vector of complex numbers or operators $X=[x_1, \ldots, x_n]^\top$, the complex conjugate or adjoint operator of $X$ is denoted by $X^\#=[x_1^\ast, \ldots, x_n^\ast]^\top$. Let $X^\dagger=(X^\#)^\top$. Clearly, when $n=1$, $X^\dag=X^\ast$.  Let $t_0$ be the initial time, i.e., the time when the system and its input start interaction. $|g\rangle$ and $|e\rangle$ stand for the ground and excited states of a two-level atom, respectively. The commutator between two operators $A$ and $B$ is $[A,B]=AB-BA$. Define super-operators
 $ \mathcal{L}X\triangleq-\mathrm{i}[X,H]+\mathcal{D}_LX$ and  $ \mathcal{L}^\star\rho\triangleq-\mathrm{i}[H,\rho]+\mathcal{D}_L^\star\rho$,
%\begin{equation*}\begin{aligned}
%&\mathrm{Lindbladian}: \mathcal{L}X\triangleq-\mathrm{i}[X,H]+\mathcal{D}_LX, \\
%&\mathrm{Liouvillian}: \mathcal{L}^\star\rho\triangleq-\mathrm{i}[H,\rho]+\mathcal{D}_L^\star\rho,
%\end{aligned}\end{equation*}
where $\mathcal{D}_LX=L^\dagger XL-\frac{1}{2}L^\dagger LX-\frac{1}{2}XL^\dagger L$ and  $\mathcal{D}_L^\star\rho=L\rho L^\dagger-\frac{1}{2}L^\dagger L\rho-\frac{1}{2}\rho L^\dagger L$.
%\begin{equation*}\begin{aligned}
%&\mathcal{D}_LX=L^\dagger XL-\frac{1}{2}L^\dagger LX-\frac{1}{2}XL^\dagger L, \\
%&\mathcal{D}_L^\star\rho=L\rho L^\dagger-\frac{1}{2}L^\dagger L\rho-\frac{1}{2}\rho L^\dagger L.
%\end{aligned}\end{equation*}

%%%%%%%%%%%%%%%%%%%%%%%%%%
%%%%%%%%%%%%%%%%%%%%%%%%%%
%%%%%%%%%%%%%%%%%%%%%%%%%%
\section{Preliminaries}\label{sec:pre}

%This section presents necessary preliminaries including quantum systems and fields, and continuous-mode single-photon states.

%%%%%%%%%%%%%%%%%%%%%%%%%%
%%%%%%%%%%%%%%%%%%%%%%%%%%
\subsection{Quantum systems  and fields} \label{subsec:sys+field}

%\begin{figure}[htp!]
%\centering
%\includegraphics[scale=0.22]{system_new.png}
%\caption{Quantum system $G$ with $m$ input and output channels.}
%\label{system_new}
%\end{figure}

Consider a  quantum system driven by $m$ input fields. The inputs  are optical or microwave fields, which are represented by the annihilation operators $b_{\mathrm{in}, k}(t)$ and their adjoints $b_{\mathrm{in}, k}^\ast(t)$ (creation operators),  $k=1,\ldots,m$. If there are no photons in an input channel, this input is in the vacuum state $\ket{\Phi_0}$. Annihilation and creation operators satisfy% singular commutation relations
%$[b_{\mathrm{in}, j}(t),b_{\mathrm{in}, k}(r)]=[b_{\mathrm{in}, j}^\ast(t),b_{\mathrm{in}, k}^\ast(r)]=0$ and $[b_{\mathrm{in}, j}(t),b_{\mathrm{in}, k}^\ast(r)]=\delta_{jk}\delta(t-r)$ for all $j,k=1,\ldots,m$ and $t,r\in\mbb{R}$.
\beq\label{eq:julu2_ccr}\begin{aligned}
&b_{\mathrm{in}, j}(t)\ket{\Phi_0}=0,\  [b_{\mathrm{in}, j}(t),b_{\mathrm{in}, k}(r)]=[b_{\mathrm{in}, j}^\ast(t),b_{\mathrm{in}, k}^\ast(r)]=0,
\\
&[b_{\mathrm{in}, j}(t),b_{\mathrm{in}, k}^\ast(r)]=\delta_{jk}\delta(t-r),  \forall j,k=1,\ldots,m, \ t,r\in\mbb{R}.
\end{aligned}\eeq
The integrated input annihilation and creation processes are respectively $B_{\mathrm{in}, k}(t)=\int_{-\infty}^t b_{\mathrm{in}, k}(r)dr$ and $B_{\mathrm{in}, k}^\ast(t)=\int_{-\infty}^t b_{\mathrm{in}, k}^\ast(r)dr$,
%\begin{eqnarray*}\label{eq:gauge}
%B_{\mathrm{in}, k}(t)=\int_{-\infty}^t b_{\mathrm{in}, k}(r)dr, ~ B_{\mathrm{in}, k}^\ast(t)=\int_{-\infty}^t b_{\mathrm{in}, k}^\ast(r)dr,
%\end{eqnarray*}
which are quantum Wiener processes. Define It\^{o} increments $dB_{\mathrm{in}, k}(t)=B_{\mathrm{in}, k}(t+dt)-B_{\mathrm{in}, k}(t)$. Denote $B_{\mathrm{in}}(t) = [B_{\mathrm{in}, 1}(t), \ldots, B_{\mathrm{in}, m}(t)]^\top$.  In this paper, the input fields are assumed to be canonical fields which include the vacuum, coherent, single- and multi-photon fields.  Then $
dB_{\mathrm{in}}(t)dB_{\mathrm{in}}^\top(\tau) = dB_{\mathrm{in}}^\#(t) dB_{\mathrm{in}}^\dag (\tau)=dB_{\mathrm{in}}^\#(t) dB_{\mathrm{in}}^\top(\tau)=\mathbf{0}$ and $dB_{\mathrm{in}}(t)dB_{\mathrm{in}}^\dag(\tau) =  I\delta_{t\tau} dt$.

The quantum system  can be parametrized by a triple $(S,L,H)$ \cite{GJ09,CKS17}. Here, $H$ is the inherent system Hamiltonian, $L = [L_1,  \ldots, L_m]^\top$ describes how the system is coupled to its environment, and $S$ is a scattering operator (e.g., a beamsplitter or a phase shifter). In this paper, it is assumed that $S=I$ (the identity operator).  The temporal evolution of the quantum system is governed by a unitary operator $U(t,t_0)$, which is the solution to the following It\^{o} quantum stochastic differential equation (QSDE):
\begin{equation*}\label{dU}\begin{aligned}
&d U(t,t_0) =\left[-(\mathrm{i}H+\ff{1}{2}L^\dag L)dt \right. \\
&\hspace{14ex}\left. +d B_{\mathrm{in}}^\dag(t) L - L^\dag
d B_{\mathrm{in}}(t)\right] U(t,t_0)
\end{aligned}\end{equation*}
under the initial condition $U(t_0,t_0)=I$. Denote the joint system-field state by $\ket{\Psi(t)}$. In the Schr\"odinger picture, $\ket{\Psi(t)}=U(t,t_0)|\Psi(t_0)\rangle$,
%\begin{equation*}
%\ket{\Psi(t)}=U(t,t_0)|\Psi(t_0)\rangle, ~ t\geq t_0,
%\end{equation*}
which is the solution to the stochastic Schr\"odinger equation
\begin{equation}\label{SME}
d\ket{\Psi(t)} =\left[-\mathrm{i} H_{\rm eff} dt +d B_{\mathrm{in}}^\dag(t) L -L^\dag
d B_{\mathrm{in}}(t)\right] \ket{\Psi(t)},
\end{equation}
where $H_{\rm eff} =  H-\ff{\mathrm{i}}{2}L^\dag L$
%\begin{equation*}\label{eq: H eff}
%H_{\rm eff} \triangleq  H-\ff{\mathrm{i}}{2}L^\dag L
%\end{equation*}
is the effective Hamiltonian that is not self-adjoint \cite[Chapter 11]{GZ04}. In particular, if the input fields are initially in the vacuum state $\ket{\Phi_0}$, then \eqref{SME} reduces to, \cite[Chapter 11]{GZ04},
\begin{equation}\label{SME_b}
d\ket{\Psi(t)} = \left(-\mathrm{i} H_{\rm eff} dt +d B_{\mathrm{in}}^\dag(t) L\right) \ket{\Psi(t)}.
\end{equation}
On the other hand, in the Heisenberg picture, the time evolution of the system operator $X$, denoted by $\mathfrak{j}_t(X)\equiv X(t)=U^\ast(t,t_0)(X\otimes I_{\mathrm{field}})U(t,t_0)$,
%\begin{equation*}\label{eq:X}
%\mathfrak{j}_t(X)\equiv X(t)=U^\ast(t,t_0)(X\otimes I_{\mathrm{field}})U(t,t_0),
%\end{equation*}
follows the It\^{o} QSDE
\begin{equation*}\begin{aligned}
&d\mathfrak{j}_t(X)=\mathfrak{j}_t(\mathcal{L}X)dt+\sum_{k=1}^m dB_{\mathrm{in},k}^\ast(t)\mathfrak{j}_t([X,L_k])
\\
&\hspace{8ex} +\sum_{k=1}^m \mathfrak{j}_t([L_k^\ast,X])dB_{\mathrm{in},k}(t).
\end{aligned}\end{equation*}
Finally, the output field annihilation operators  are $B_{\mathrm{out},k}(t) = U^\ast(t,t_0) B_{\mathrm{in},k}(t)  U(t,t_0)$, $(k=1,\ldots,m)$,
%\begin{equation*}
%B_{\mathrm{out},k}(t) = U^\ast(t,t_0) B_{\mathrm{in},k}(t)  U(t,t_0), ~~~ k=1,\ldots,m,
%\end{equation*}
whose dynamical evolution is $dB_{\mathrm{out}}(t)=L(t)dt+dB_\mathrm{in}(t)$.
%\begin{eqnarray*}\label{B_out}
%dB_{\mathrm{out}}(t)=L(t)dt+dB_\mathrm{in}(t).
%\end{eqnarray*}
More discussions on open quantum systems can be found in, e.g., \cite{GZ04,WM09,GJ09,CKS17}.

%%%%%%%%%%%%%%%%%%%%%%%%%%
%%%%%%%%%%%%%%%%%%%%%%%%%%
\subsection{Continuous-mode single-photon states}\label{subsec:eq:photon}

For each input channel $k=1,\ldots,m$, the creation operator $b_{\mathrm{in}, k}^\ast$  generates a photon from the vacuum. Mathematically, $|1_{k,t}\rangle \triangleq b_{\mathrm{in}, k}^\ast (t)\ket{\Phi_0}$ means a photon is generated at time $t$ in the $k$th input channel.  By  \eqref{eq:julu2_ccr},  $\langle1_{j,t}|1_{k,\tau}\rangle =\delta_{jk}\delta(t-\tau)$. Hence $\{|1_{k,t}\rangle: t\in\mbb{R}\}$ is an orthogonal basis of single-photon states for each channel $k$. Indeed, a  single-photon state with temporal pulse shape $\xi(t)$ in the $k$th channel can be viewed as a superposition of a continuum of $|1_{k,t}\rangle$, i.e., %$|1_{\xi}\rangle = \int_{-\infty}^\infty \xi(t) |1_{k,t}\rangle dt = \int_{-\infty}^\infty \xi(t) b_{\mathrm{in}, k}^\ast(t)dt\ket{\Phi_0}$.
\begin{equation}\la{1 photon jan21}
|1_{\xi}\rangle = \int_{-\infty}^\infty \xi(t) |1_{k,t}\rangle dt = \int_{-\infty}^\infty \xi(t) b_{\mathrm{in}, k}^\ast(t)dt\ket{\Phi_0}.
\end{equation}
Physically, as a quantum state, $|1_{\xi}\rangle$ can be interpreted in the following way: the probability of finding the photon in the time bin $[t,t+dt)$ is $|\xi(t)|^2 dt$. The normalization condition $\langle1_{\xi}|1_{\xi}\rangle=1$ requires $ \int_{-\infty}^\infty |\xi(t)|^2 dt=1$.  As the single-photon state $|1_{\xi}\rangle$ is parameterized by an $L^2$ integrable function $\xi(t)$ over $\mbb{C}$, it is called a {\it continuous-mode} single-photon state \cite{milburn08,FKS10,GJNC12,ZJ13,GZ15b,PZJ16,PDZ16,DTK+18}.
If a single photon is superposed over $m$ channels, then the single-photon state is a superposition state of the form %$|1_\zeta\rangle = \sum_{k=1}^m \int_{-\infty}^\infty \zeta_k(t) |1_{k,t}\rangle dt$,
\begin{equation*}\label{eq:1 photon}
|1_\zeta\rangle = \sum_{k=1}^m \int_{-\infty}^\infty \zeta_k(t) |1_{k,t}\rangle dt,
\end{equation*}
whose normalization condition is $\sum_{k=1}^m \int_{-\infty}^\infty |\zeta_k(t)|^2 dt=1$. In particular, if $\zeta_j(t)\equiv 0$ for some $j=1,\ldots, m$, then it means that the input channel $j$ is in the vacuum state and the photon is superposed over the other input channels.
%References on single-photon states?

%%%%%%%%%%%%%%%%%%%%%%%%%%
%%%%%%%%%%%%%%%%%%%%%%%%%%
%%%%%%%%%%%%%%%%%%%%%%%%%%
\section{The Tavis-Cummings model}\label{sec:system}

We first present the Tavis-Cummings model in subsection \ref{sec:des}.  In subsection \ref{subsec:linear}, assuming that the atoms are initially in the ground state, the cavity is empty  and the input field is in the vacuum state, we present an associated linear model. Physical interpretation of the linear model is discussed in subsection \ref{subsec:physical}.
%?????
%The steady-state output field state is derived in this section. As a special case, the photon probability distribution for $N \leq 4$ with respect to various system parameters are also simulated.

%%%%%%%%%%%%%%%%%%%%%%%%%%
%%%%%%%%%%%%%%%%%%%%%%%%%%
\subsection{The Tavis-Cummings model}\label{sec:des}

\begin{figure}[htp!]
\centering
\includegraphics[scale=0.2]{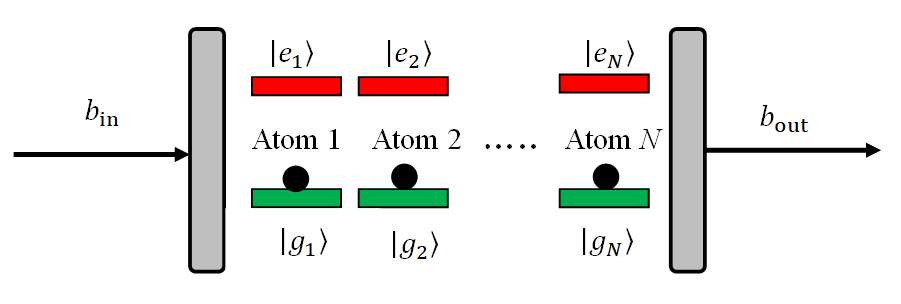}
\caption{Schematic of the Tavis-Cummings model. $N$ two-level atoms are coupled to a single-mode  cavity, which is driven by an input field. (Here, the atoms are inside the cavity. However, in some physical implementations atoms may be outside the cavity, but of course they must couple to the cavity; see e.g., \cite[Fig. 1]{Van18}). The cavity can also be a transmission line in a superconducting platform.}
\label{system}
\end{figure}

In the Tavis-Cummings model as shown in Fig. \ref{system}, the $N$ TLSs are not directly coupled to each other; instead, they all couple to the common single-mode cavity. The inherent system Hamiltonian of the Tavis-Cummings model is (\cite[(2.1)]{TC68}, \cite[(1)]{SPS07}, \cite[(1)]{MCG07},\cite[(1)]{Fink09},\cite[(C7)]{Van18},\cite[(1)]{WW+PRL20})
\begin{equation}\label{H_TC}
H_{\mathrm{TC}}=\omega_r a^\ast a + \sum_{j=1}^{N}\left[\frac{\omega_j}{2}\sigma_{z,j}+\Gamma_j(a^\ast\sigma_{-,j}+\sigma_{+,j}a)\right].
\end{equation}
Here, $a$, $a^\ast$ denote the annihilation and creation operators of the cavity mode satisfying $[a,a^\ast]=I$, $\omega_r$ is the frequency detuning between the cavity mode and the input field. The two-level atom $j$ is coupled to the cavity with coupling strength $\Gamma_j$, $j=1,\ldots,N$, which is assumed to a real number but can be negative \cite[Table 1]{Fink09} or \cite[(2)]{ANP+17}. The corresponding detuning between the transition frequency of the two-level atom $j$ and the carrier frequency of the input field is denoted by $\omega_j$.  The lowering and raising operators of the two-level atom $j$ are $\sigma_{-,j}= \ket{g_j}\bra{e_j}$ and $\sigma_{+,j} = \ket{e_j}\bra{g_j}$, respectively.  The Pauli $Z$ operator is $\sigma_{z,j}=\sigma_{+,j}\sigma_{-,j}-\sigma_{-,j}\sigma_{+,j}$.   %he central frequency of the inout field is set to unit $1$ throughout this paper. ?????
The system exchanges information with its environment by means of absorbing and emitting photons, which is realized by the coupling operator $L=\sqrt{\kappa}a$.  By the development in subsection \ref{subsec:sys+field}, the  It\^{o} QSDEs for the Tavis-Cummings model in the Heisenberg picture are
%\begin{equation}\label{system1}
%\left\{\begin{aligned}
%\dot{\sigma}_{-,1}=&-\mathrm{i}[\sigma_{-,1},H_{\mathrm{TC}}]=\mathrm{i}\omega_1\sigma_{-,1}-\mathrm{i}g_1\sigma_{z,1}a, \\
%%\dot{\sigma}_{-,2}=&-\mathrm{i}[\sigma_{-,2},H_{\mathrm{TC}}]=\mathrm{i}\omega_2\sigma_{-,2}-\mathrm{i}g_2\sigma_{z,2}a, \\
%\vdots& \\
%\dot{\sigma}_{-,N}=&-\mathrm{i}[\sigma_{-,N},H_{\mathrm{TC}}]=\mathrm{i}\omega_N\sigma_{-,N}-\mathrm{i}g_N\sigma_{z,N}a, \\
%\dot{a}=&-\mathrm{i}[a,H_{\mathrm{TC}}]+\mathcal{D}_La-\sqrt{\kappa}b_{\mathrm{in}} \\
%=&-\mathrm{i}\omega_ra-\mathrm{i}\sum_{j=1}^{N}g_j\sigma_{-,j}-\frac{\kappa}{2}a-\sqrt{\kappa}b_{\mathrm{in}}, \\
%b_{\mathrm{out}}=&\sqrt{\kappa}a+b_{\mathrm{in}},
%\end{aligned}\right.
%\end{equation}
\begin{equation}\label{system1}
\left\{\begin{aligned}
d\sigma_{-,1}(t)&=-\mathrm{i}\omega_1\sigma_{-,1}(t)dt+\mathrm{i}\Gamma_1\sigma_{z,1}(t)a(t)dt, \\
%\dot{\sigma}_{-,2}=&-\mathrm{i}[\sigma_{-,2},H_{\mathrm{TC}}]=\mathrm{i}\omega_2\sigma_{-,2}-\mathrm{i}g_2\sigma_{z,2}a, \\
&\ \vdots \\
d\sigma_{-,N}(t)&=-\mathrm{i}\omega_N\sigma_{-,N}(t)dt+\mathrm{i}\Gamma_N\sigma_{z,N}a(t)dt, \\
da(t)&=-(\mathrm{i}\omega_r+\frac{\kappa}{2})a(t)dt
\\
&\hspace{2ex}-\mathrm{i}\sum_{j=1}^{N}\Gamma_j\sigma_{-,j}(t)dt-\sqrt{\kappa}dB_{\mathrm{in}}(t), \\
dB_{\mathrm{out}}(t)&=\sqrt{\kappa}a(t)dt+dB_{\mathrm{in}}(t), \ t\geq t_0,
\end{aligned}\right.
\end{equation}
which is  bilinear.

%%%%%%%%%%%%%%%%%%%%%%%%%%
\begin{remark}
When all atoms are resonant with the resonator, i.e., $\omega_1=\cdots=\omega_N = \omega_r \equiv\omega_s$,  rotations $\sigma_{-,j}(t)\to e^{\i \omega_s t}\sigma_{-,j}(t)$, $a(t)\to e^{\i \omega_s t}a(t)$, $B_{\mathrm{in}}(t)\to e^{\i \omega_s t}B_{\mathrm{in}}(t)$, and $B_{\mathrm{out}}(t)\to e^{\i \omega_s t}B_{\mathrm{out}}(t)$ convert system \eqref{system1} to
\begin{equation}\label{system1_june28}
\left\{\begin{aligned}
d\sigma_{-,1}(t)&=\mathrm{i}\Gamma_1\sigma_{z,1}(t)a(t)dt, \\
\vdots& \\
d\sigma_{-,N}(t)&=\mathrm{i}\Gamma_N\sigma_{z,N}a(t)dt, \\
da(t)&=-\frac{\kappa}{2}a(t)dt-\mathrm{i}\sum_{j=1}^{N}\Gamma_j\sigma_{-,j}(t)dt-\sqrt{\kappa}dB_{\mathrm{in}}(t), \\
dB_{\mathrm{out}}(t)&=\sqrt{\kappa}a(t)dt+dB_{\mathrm{in}}(t), \ t\geq t_0.
\end{aligned}\right.
\end{equation}
In other words, when all the atoms and the cavity are resonant, the frequencies do not affect the system dynamics, except the central frequency of the radiation field.
\end{remark}

%%%%%%%%%%%%%%%%%%%%%%%%%%
%%%%%%%%%%%%%%%%%%%%%%%%%%
\subsection{The corresponding linear model}\label{subsec:linear}
%?????

Assume that all the two-level atoms are initially in the ground state and  the cavity is in the vacuum state $\ket{0}$. That is, the initial state of the Tavis-Cummings model is
\begin{equation}\label{feb3:initial}
|\zeta\rangle=|g_1g_2\cdots g_N\rangle\otimes|0\rangle.
\end{equation}
Let
%$X(t)=[ \sigma_{-,1}(t) , \ldots,  \sigma_{-,N}(t), a(t)]^\top$.
\begin{equation} \label{eq:jun4_X}
X(t)=\left[
       \begin{array}{ccccc}
         \sigma_{-,1}(t) & \sigma_{-,2}(t) & \cdots & \sigma_{-,N}(t) & a(t) \\
       \end{array}
     \right]^\top.
\end{equation}
Notice that
\beq\label{eq:z_g}
\sigma_{z,j}\ket{\zeta\Phi_0}=-\ket{\zeta\Phi_0}
\eeq
 for all $j=1,\ldots,N$. From \eqref{system1} we get
%\begin{equation}\label{system2}\begin{aligned}
%\bra{\zeta\Phi_0}\dot{X}(t)&=A\bra{\zeta\Phi_0}X(t)+B\bra{\zeta\Phi_0}b_{\mathrm{in}}(t), \\
%\bra{\zeta\Phi_0}b_{\mathrm{out}}(t)&=C\bra{\zeta\Phi_0}X(t)+\bra{\zeta\Phi_0}b_{\mathrm{in}}(t),
%\end{aligned}\end{equation}
\begin{equation}\label{system2}\begin{aligned}
 dX(t)\ket{\zeta\Phi_0}&=AX(t)\ket{\zeta\Phi_0} dt+B dB_{\mathrm{in}}(t)\ket{\zeta\Phi_0}, \\
 dB_{\mathrm{out}}(t)\ket{\zeta\Phi_0}&=CX(t)\ket{\zeta\Phi_0}dt+dB_{\mathrm{in}}(t)\ket{\zeta\Phi_0},
\end{aligned}\end{equation}
where
\begin{equation}\label{ABC}\begin{aligned}
&A=-\mathrm{i}\left[
    \begin{array}{ccccc}
      \omega_1 & 0 & \cdots & 0 & \Gamma_1 \\
      0 & \omega_2 & \cdots & 0 & \Gamma_2 \\
      \vdots & \vdots & \ddots & \vdots & \vdots \\
      0 & 0 & \cdots & \omega_N & \Gamma_N \\
     \Gamma_1 & \Gamma_2 & \cdots & \Gamma_N &\omega_r-\frac{\kappa\mathrm{i}}{2}
    \end{array}
  \right], \\
&B=\left[
    \begin{array}{ccccc}
      0 & 0 & \cdots & 0 & -\sqrt{\kappa}
    \end{array}
  \right]^\top,~~
C=-B^\top.
\end{aligned}\end{equation}
\eqref{system2} is a linear system. Actually, a linear quantum system of $N+1$ quantum harmonic oscillators $\bar{a}= [a_1, \ldots, a_N,  a]^\top$  with system Hamiltonian $H=\omega_r a^\ast a + \sum_{j=1}^{N}\left[\omega_ja_j^\ast a_j+\Gamma_j(a^\ast a_j+a_j^\ast a)\right]$
%\begin{equation*}\label{H_TC2}
%H=\omega_r a^\dagger a + \sum_{j=1}^{N}\left[\omega_ja_j^\dag a_j+g_j(a^\dagger a_j+a_j^\dag a)\right]
%\end{equation*}
and coupling operator $L=\sqrt{\kappa}a$ has the following linear It\^o QSDEs
%\begin{equation}\label{linear_system}\begin{aligned}
%\dot{a}(t)&=Aa(t)+Bb_{\mathrm{in}}(t), \\
%b_{\mathrm{out}}(t)&=Ca(t)+ b_{\mathrm{in}}(t),
%\end{aligned}\end{equation}
\begin{equation}\label{linear_system}\begin{aligned}
d\bar{a}(t)&=A\bar{a}(t)dt+BdB_{\mathrm{in}}(t), \\
dB_{\mathrm{out}}(t)&=C\bar{a}(t)dt+ dB_{\mathrm{in}}(t),
\end{aligned}\end{equation}
where $A,B,C$ are exactly those in  \eqref{ABC}. More discussions on linear quantum systems theory can be found in e.g., \cite{GZ15,NY17,ZGPG18}.  %we have
%\begin{equation}\label{out1}\begin{aligned}
%\langle\Phi\Phi_0|b_{\mathrm{out}}(t)=&Ce^{A(t-0)}\langle\Phi\Phi_0|X(0) \\
%&+\int_{0}^{t}g_G(t-\tau)\langle\Phi\Phi_0|b_{\mathrm{in}}(\tau)\ d\tau,
%\end{aligned}\end{equation}
%where $0$ is the initial time of the input field interacting with the Tavis-Cummings model.
%the impulse response function of the system \eqref{linear_system}   is
%\begin{equation}\label{impulse}
%g_G(t)=\left\{\begin{array}{cc}
%                \delta(t)+Ce^{At}B, & t\geq0, \\
%                0, & t<0,
%              \end{array}
%\right.
%\end{equation}
%and the corresponding
The transfer function of the linear quantum system \eqref{linear_system} is
\begin{equation}\label{tran}
G[s]=1+C(sI-A)^{-1}B.
\end{equation}
Plugging \eqref{ABC} into \eqref{tran} yields
\begin{equation}\label{transfer}\begin{aligned}
G[s]
=\frac
{\displaystyle{\sum_{k=1}^N}\left(\Gamma_k^2 +\frac{1}{N}(s+\mathrm{i}\omega_k)(s+\mathrm{i}\omega_r-\frac{\kappa}{2})\right) \prod_{j\neq k}^{N}(s+\mathrm{i}\omega_j)}
{\displaystyle{\sum_{k=1}^N}\left(\Gamma_k^2 +\frac{1}{N}(s+\mathrm{i}\omega_k)(s+\mathrm{i}\omega_r+\frac{\kappa}{2})\right) \prod_{j\neq k}^{N}(s+\mathrm{i}\omega_j)}.
\end{aligned}\end{equation}
%\begin{equation}\label{transfer}\begin{aligned}
%&G[s] \\
%=&\Bigg[2g_1^2\prod_{j\neq1}^N(s-\mathrm{i}\omega_j)+2g_2^2\prod_{j\neq2}^{N}(s-\mathrm{i}\omega_j)+\cdots \\
%&+2g_N^2\prod_{j\neq N}^{N}(s-\mathrm{i}\omega_j)+(2s+2\mathrm{i}\omega_r-\kappa)\prod_{j=1}^{N}(s-\mathrm{i}\omega_j)\Bigg]\Bigg/ \\
%&\Bigg[2g_1^2\prod_{j\neq1}^N(s-\mathrm{i}\omega_j)+2g_2^2\prod_{j\neq2}^{N}(s-\mathrm{i}\omega_j)+\cdots \\
%&+2g_N^2\prod_{j\neq N}^{N}(s-\mathrm{i}\omega_j)+(2s+2\mathrm{i}\omega_r+\kappa)\prod_{j=1}^{N}(s-\mathrm{i}\omega_j)\Bigg].
%\end{aligned}\end{equation}
%By the explicit form of the transfer function \eqref{transfer}, it can be verified that  $G[\mathrm{i}\omega]^\ast G[\mathrm{i}\omega]=1$ for all $\omega$. As a result, it only modulates the phase of its input signal.
If $\omega_1=\cdots=\omega_N\equiv\omega_s$,  \eqref{transfer} reduces to
\begin{equation}\label{eq:transfer_3}
G[s]=\frac
{(\sqrt{N}\bar{\Gamma})^2 + (s+\mathrm{i}\omega_s)(s+\mathrm{i} \omega_r -\frac{\kappa}{2})}
{(\sqrt{N}\bar{\Gamma})^2 + (s+\mathrm{i}\omega_s)(s+\mathrm{i} \omega_r +\frac{\kappa}{2})},
\end{equation}
where $\bar{\Gamma}\triangleq\sqrt{\frac{1}{N}\sum_{k=1}^N \Gamma_k^2}$.
%Moreover, if $\omega_r=\omega_s=0$, i.e., the cavity and the atoms are all  resonant to the input field, then the poles of the transfer function \eqref{eq:transfer_3} are $s=-\frac{\kappa}{4} \pm \mathrm{i} \varpi$,
%where $\varpi = \sqrt{N\bar{g}^2-\frac{\kappa^2}{16}}$.
%%In this case, the impulse response function is
%%\begin{equation}\label{eq:oct25_impulse}
%%g_G(t)=\delta(t)-\kappa e^{-\frac{\kappa}{4}t}\left(\cos(\varpi t)-\frac{\kappa}{4\varpi}\sin(\varpi t)\right).
%%\end{equation}
%%Thus, $\varpi$ plays the role of the Rabi frequency.
%We have
%%\begin{equation}
%%\lim_{\kappa\rightarrow0}\varpi = \sqrt{N}\bar{g},
%%\end{equation}
%\begin{equation} \label{eq:jun4_poles}
%\lim_{\kappa\rightarrow0}s = \pm\mathrm{i}\sqrt{N}\bar{g}.
%\end{equation}
Let $\omega_r=\omega_s=0$, i.e., all atoms are resonant with the cavity resonator. Define  $T[s] \triangleq G[s]-1$. Then
%\beq\label{eq:T_s}
% = \frac{-\kappa s}{(\sqrt{N}\bar{g})^2 +s(s+\frac{\kappa}{2})}.
%\eeq
\beq\la{eq:T_s2}
|T[\i\omega]|^2 = \ff{\kappa^2 \omega^2}{[(\sqrt{N}\bar{\Gamma})^2-\omega^2]^2+\ff{\kappa^2}{4}\omega^2}.
\eeq
Clearly, $|T[0]|^2=0$ and $|T[\i\omega]|^2$ has two peaks attained at $\omega=\pm \sqrt{N}\bar{\Gamma}$, respectively. (If $N=0$, then $|T[\i\omega]|^2$ has only one peak attained at $\omega=0$, which is the empty cavity case.)

%%%%%%%%%%%%%%%%%%%%%%%%%%
\bmrk\label{rem:passivity}
Suppose energy enters the system via the input field $B_{\rm{in}}$ and flows out through the output field $B_{\rm{out}}$. $T[s]$ is related to the energy stored in the system. In fact, define the performance variable $z\triangleq C\bar{a}$. Then $T[s]$ is the transfer function from $B_{\rm{in}}$ to $z$.  It is easy to see that
\begin{equation}
\left[
\begin{array}{cc}
A + A^\dag  + C^\dag C  & B -C^\dag
\\
B^\dag  - C  & 0
\end{array}
 \right ] = 0.
\label{eq:prl-111}
\end{equation}
Thus, by the positive real lemma in \cite[Theorem 3]{ZJ11}, the system \eqref{linear_system} with the performance variable $z$ is passive. In subsection \ref{subsec:physical}, we will show that $T[s]$ also reflects the vacuum Rabi mode splitting of the Tavis-Cummings system  \eqref{system1}.  Hence, $T[s]$  bridges the passivity of the linear quantum system  \eqref{linear_system} and the  vacuum Rabi mode splitting phenomenon exhibited by the  Tavis-Cummings system  \eqref{system1}.
\emrk
%Denote $\bra{\Phi\Phi_0}X(t)$ in system \eqref{linear_system} by
%\begin{equation}\label{eq:new x}
%\bra{\Phi\Phi_0}X(t)
%=\left[
%       \begin{array}{ccccc}
%         a_1(t) & a_2(t) & \cdots & a_N(t) & a(t) \\
%       \end{array}
%     \right]^\top.
%\end{equation}
%In contrast to \eqref{H_TC}, the Hamiltonian for a linear quantum system is
%\begin{equation}\label{H_TC2}
%H=\omega_r a^\dagger a + \sum_{j=1}^{N}\left[\omega_ja_j^\dag a_j+g_j(a^\dagger a_j+a_j^\dag a)\right].
%\end{equation}

In the following, we perform structural decomposition on the linear quantum system \eqref{linear_system}. We partition $\omega_1,\ldots,\omega_N$ into $M$ groups according to their degeneracies, where the $j$th degenerated frequency is denoted by  $\tilde{\omega}_j$, ($j=1,\ldots, M$).  Let the number of elements be $n_j$ for the group $j$.  In particular, if $M=N$, then $\omega_j\neq \omega_k$ for all $1\leq j<k\leq N$. For convenience, we can arrange the elements of $\bar{a}(t)$ in \eqref{linear_system} so that the matrix $A$ in \eqref{ABC} is of the form

\scriptsize
\vspace{-2ex}
\begin{align*}
&\tilde{A}\\
=&-\mathrm{i}\left[
            \begin{array}{cccccccc}
             \tilde{\omega}_1 & \cdots & 0 &\cdots & 0 & \cdots & 0 & \Gamma_{1_1} \\
              \vdots & \ddots & \vdots & \vdots &\vdots & \ddots & \vdots & \vdots \\
              0 & \cdots &\tilde{\omega}_1 &\cdots & 0 & \cdots & 0 & \Gamma_{1_{n_1}} \\
              \vdots & \ddots &  \vdots &\ddots & \vdots  & \ddots & \vdots & \vdots \\
              0 & \cdots & 0 &\cdots &\tilde{\omega}_M & \cdots & 0 & \Gamma_{M_1} \\
               \vdots & \ddots &  \vdots &\ddots & \vdots  & \ddots & \vdots &  \vdots\\
              0 & \cdots & 0 &\cdots  & 0 & \cdots &\tilde{\omega}_M & \Gamma_{M_{n_M}} \\
             \Gamma_{1_1} & \cdots & \Gamma_{1_{n_1}} &\cdots & \Gamma_{M_1} & \cdots & \Gamma_{M_{n_M}} &\omega_r-\frac{\kappa \mathrm{i}}{2} \\
            \end{array}
          \right].
\end{align*}
\normalsize
In other words, we group the elements of $\bar{a}(t)$ according to the partition of the detuned frequencies. It is easy to see that matrices $B$ and $C$ remain the same under this re-arrangement.
%Thus, there are $\omega_{j_1},\ldots,\omega_{j_{n_j}}$ which are all equal to $\tilde{\omega}_j$, $1\leq j_1<\cdots<j_{n_j}\leq N$.

%We have the following lemma.

%%%%%%%%%%%%%%%%%%%%%%%%%%
\begin{lemma}\label{lem:kalman}
Partition $\omega_1,\ldots,\omega_N$ into groups as described above. There is
an orthogonal matrix $\tilde{T}$ that transforms the linear quantum system \eqref{linear_system} with system matrices $(\tilde{A},B,C)$ to another one, denoted $\Sigma$, with system matrices
%\vspace{-5.5ex}
\begin{equation}\label{May28-4}\begin{aligned}
&\hat{A}=\tilde{T}^\top \tilde{A} \tilde{T} \\
=&-\mathrm{i}\left[
  \begin{array}{cccccc}
 \tilde{\omega}_1I_{n_1-1} & \mathbf{0} & \cdots & \cdots & \cdots & \mathbf{0} \\
    \mathbf{0} & \tilde{\omega}_1 & \cdots & \cdots & \cdots & \sqrt{\tilde{\Gamma}_1} \\
    \vdots & \vdots & \ddots & \cdots & \cdots & \vdots \\
    \vdots & \vdots & \vdots & \tilde{\omega}_MI_{n_M-1} & \cdots & \mathbf{0} \\
    \vdots & \vdots & \vdots & \vdots & \tilde{\omega}_M & \sqrt{\tilde{\Gamma}_M} \\
    \mathbf{0} & \sqrt{\tilde{\Gamma}_1} & \cdots & \mathbf{0} & \sqrt{\tilde{\Gamma}_M} & \omega_r-\frac{\kappa\mathrm{i}}{2} \\
  \end{array}
\right], \\
&\hat{B}=\tilde{T}^\top B = B ,~~\hat{C}= C\tilde{T} = C,
\end{aligned}\end{equation}
where $\tilde{\Gamma}_j\triangleq\sum_{k=1}^{n_j}\Gamma_{j_k}^2$, $(j=1,\ldots, M)$.
\end{lemma}

%%%%%%%%%%%%%%%%%%%%%%%%%%
\begin{proof}
The proof is constructive. Define a matrix $\tilde{T}$ as
%\vspace{-1.0ex}
\begin{equation} \label{eq:tilde_T}
\tilde{T}=\left[
            \begin{array}{ccccc}
              \tilde{T}_1 & \mathbf{0}  & \cdots &  \mathbf{0} &  \mathbf{0} \\
               \mathbf{0} & \tilde{T}_2 & \cdots &  \mathbf{0} &  \mathbf{0} \\
              \vdots & \vdots & \ddots & \vdots & \vdots \\
               \mathbf{0} &  \mathbf{0} & \cdots   & \tilde{T}_M &  \mathbf{0} \\
               0&  0 & \cdots &  0 & 1 \\
            \end{array}
          \right],
\end{equation}
where
for each $j=1,\ldots,M$,  $\tilde{T}_j=[
                T_{j_1} \ T_{j_2}  \ \cdots \ T_{j_{n_j}}] \in \mathbb{R}^{n_j\times n_j}$,
%\begin{equation*}\label{mar4-1}\begin{aligned}
%\tilde{T}_j=\left[
%              \begin{array}{cccc}
%                T_{j_1} & T_{j_2} & \cdots & T_{j_{n_j}}
%              \end{array}
%            \right]\in \mathbb{R}^{n_j\times n_j}, ~~ 1\leq j \leq M,
%\end{aligned}\end{equation*}
in which
\begin{equation*}\label{May28-1}\begin{aligned}
T_{j_1}=&\ \sqrt{\frac{\Gamma_{j_2}^2}{\Gamma_{j_1}^2+\Gamma_{j_2}^2}}\left[1~\frac{-\Gamma_{j_1}}{\Gamma_{j_2}}~0~\cdots~0\right]^\top, \\
T_{j_2}=&\ \sqrt{\frac{\Gamma_{j_2}^2\Gamma_{j_3}^2}{(\Gamma_{j_1}^2+\Gamma_{j_2}^2)(\Gamma_{j_1}^2+\Gamma_{j_2}^2+\Gamma_{j_3}^2)}} \\
&\times\left[\frac{\Gamma_{j_1}}{\Gamma_{j_2}}~1~\frac{-(\Gamma_{j_1}^2+\Gamma_{j_2}^2)}{\Gamma_{j_2}\Gamma_{j_3}}~0~\cdots~0\right]^\top, \\
\vdots& \ \\
T_{j_{(n_j-1)}}=&\ \sqrt{\frac{\Gamma_{j_{(n_j-1)}}^2\Gamma_{j_{n_j}}^2}{(\sum_{k=1}^{n_j-1}\Gamma_{j_k}^2)(\sum_{k=1}^{n_j}\Gamma_{j_k}^2)}} \\
&\times\left[\frac{\Gamma_{j_1}}{\Gamma_{j_{(n_j-1)}}}~\frac{\Gamma_{j_2}}{\Gamma_{j_{(n_j-1)}}}~\cdots~1~\frac{-\sum_{k=1}^{n_j-1}\Gamma_{j_k}^2}{\Gamma_{j_{(n_j-1)}}g_{j_{n_j}}}\right]^\top, \\
T_{j_{n_j}}=&\ \sqrt{\frac{1}{\sum_{k=1}^{n_j}\Gamma_{j_k}^2}}\left[\Gamma_{j_1}~\Gamma_{j_2}~\cdots~\Gamma_{j_{(n_j-1)}}~\Gamma_{j_{n_j}}\right]^\top.
\end{aligned}\end{equation*}
It can be easily verified that  $\tilde{T}$ is orthogonal.  Moreover, simple algebraic manipulations yield that $\tilde{T}^\top\tilde{A}\tilde{T} = \hat{A}$.
\end{proof}

The transformed linear quantum system with system matrices $(\hat{A},B,C)$  has a nice structure. Denote by $b_j$ the system coordinate corresponding to the row  whose last entry is $-\mathrm{i}\sqrt{\bar{\Gamma}_j}$ in the matrix $\hat{A}$, $j=1,\ldots, M$. Then from the structure of $(\hat{A},B,C)$ it can be easily seen that this system has a subsystem of the form
%\begin{equation}\label{system_linear}
%\left\{\begin{aligned}
%\dot{b}_1=&\ \mathrm{i}\tilde{\omega}_1 b_1-\mathrm{i}\sqrt{\bar{g}_1} a, \\
%\vdots& \\
%\dot{b}_M=&\ \mathrm{i}\tilde{\omega}_M b_M-\mathrm{i}\sqrt{\bar{g}_M} a, \\
%\dot{a}_0=&-(\mathrm{i}\omega_r+\frac{\kappa}{2})a-\mathrm{i}\sum_{j=1}^{M} \sqrt{\bar{g}_j}b_j-\sqrt{\kappa}b_{\mathrm{in}},\\
%b_{\mathrm{out}}=&\sqrt{\kappa}a+b_{\mathrm{in}}.
%\end{aligned}\right.
%\end{equation}
\begin{equation}\label{system_linear}
\left\{\begin{aligned}
db_1(t)=&\ -\mathrm{i}\tilde{\omega}_1 b_1(t)dt-\mathrm{i}\sqrt{\bar{\Gamma}_1} a(t)dt, \\
\vdots& \\
db_M(t)=&\ -\mathrm{i}\tilde{\omega}_M b_M(t)dt-\mathrm{i}\sqrt{\bar{\Gamma}_M} a(t)dt, \\
da(t)=&-(\mathrm{i}\omega_r+\frac{\kappa}{2})a(t)dt
\\
&-\mathrm{i}\sum_{j=1}^{M} \sqrt{\bar{\Gamma}_j}b_j(t)dt-\sqrt{\kappa}d B_{\mathrm{in}}(t),\\
dB_{\mathrm{out}}(t)=&\sqrt{\kappa}a(t)dt+dB_{\mathrm{in}}(t).
\end{aligned}\right.
\end{equation}
The other $M-1$ subsystems are all isolated systems, which are called decoherence-free subsystems (DFSs) in the linear quantum control literature \cite{NY14tac,GZ15,NY17,ZGPG18}. Of course, if $n_j=1$ for some $j=1,\ldots,M$, there is no such subsystem as can be seen clearly from the matrix $\hat{A}$ in \eqref{May28-4}. According to quantum linear systems theory,  the subsystem  \eqref{system_linear} is a controllable and observable subsystem.

The following result is an immediate consequence of Lemma \ref{lem:kalman}.

%%%%%%%%%%%%%%%%%%%%%%%%%%
\begin{corollary}\label{cor:kalman}
If $\omega_1=\cdots=\omega_N\equiv \omega_s$, then $M=1$ and $n_1=N$ in \eqref{May28-4}. Accordingly, the  transformation matrix $\tilde{T}$ in \eqref{eq:tilde_T} reduces to
 \begin{equation}\label{feb22-1}
\left[
    \begin{array}{c|c}
      \tilde{T}_1  & \mathbf{0} \\
      0& 1 \\
    \end{array}
  \right]\equiv T,
\end{equation}
which transforms the quantum linear system \eqref{linear_system} to a new one with system matrices
\begin{equation}\label{eq:jun4_ABC_4}\begin{aligned}
&\hat{A}=T^\top A T =\left[
                \begin{array}{c|c}
                  \hat{A}_{\bar{c}\bar{o}} & \mathbf{0} \\  \hline
                  \mathbf{0} & \hat{A}_{co} \\
                \end{array}
              \right], \\
             & \hat{B}=T^\top B=\left[
                \begin{array}{c}
                \hat{B}_{\bar{c}\bar{o}} \\  \hline
               \hat{B}_{co}  \\
                \end{array}
              \right],\hat{C}=CT=[ \hat{C}_{\bar{c}\bar{o}} \  | \  \hat{C}_{co} ],
\end{aligned}\end{equation}
where
\begin{equation}\label{feb22-2}\begin{aligned}
&\hat{A}_{\bar{c}\bar{o}}=-\mathrm{i}\omega_s I_{N-1}, \ \
\hat{A}_{co}=-\left[
         \begin{array}{cc}
           \mathrm{i}\omega_s & \mathrm{i}\sqrt{N}\bar{\Gamma} \\
           \mathrm{i}\sqrt{N}\bar{\Gamma} & \mathrm{i}\omega_r+\frac{\kappa}{2}
         \end{array}
       \right]
       \\
&\hat{B}_{\bar{c}\bar{o}}  =  \mathbf{0},       \
 \hat{B}_{co} = -\left[
                \begin{array}{c}
                  0 \\
                \sqrt{\kappa}
                \end{array}
              \right],  \hat{C}_{\bar{c}\bar{o}} =  \mathbf{0},  \
              \hat{C}_{co} =
              \left[
                  \begin{array}{cc}
               0 &  \sqrt{\kappa}
               \end{array}\right].
\end{aligned}\end{equation}
\end{corollary}

%%%%%%%%%%%%%%%%%%%%%%%%%%
\bmrk
In the linear quantum systems theory \cite{NY14tac,GZ15,NY17,ZGPG18} DFSs   are  defined in the Heisenberg picture. On the other hand, decoherence-free subspaces widely used  in the quantum information community  are characterized by  density matrices in the Schr\"odinger picture; see references 3-8 in \cite{TV08}.  It can be seen that  the eigenvectors $\{T_1,\ldots, T_{N-1}\}$ of the DFS with system matrices  $(\hat{A}_{\bar{c}\bar{o}}, \hat{B}_{\bar{c}\bar{o}}, \hat{C}_{\bar{c}\bar{o}})$ constitute  a basis of  the  decoherence-free subspace. Some examples will be given in subsection \ref{subsec:physical}.
\emrk

%%%%%%%%%%%%%%%%%%%%%%%%%%
\begin{remark}
The coupling operator $L=\sqrt{\kappa}a$ is called the noise operator in \cite{TV08}, through which the system information leaks irreversibly into its surrounding environment. It can be clearly seen from  \eqref{eq:jun4_ABC_4}-\eqref{feb22-2} that $L$ has no effect on the DFS with system matrices $(\hat{A}_{\bar{c}\bar{o}}, \hat{B}_{\bar{c}\bar{o}}, \hat{C}_{\bar{c}\bar{o}})$, thus the DFS is robust with respect to the cavity decay rate $\kappa$. This is the so-called $\gamma$-robustness in \cite[Definition 3]{TV08}. Moreover, from  \eqref{feb22-2} it can be seen that the DFS  is robust with respect to the variations of the atom-cavity coupling strengths $\Gamma_j$ as well. Finally, the DFS is not attractive from \cite[Proposition 3]{TV08}, which is for noiseless subsystems and reduces to DFSs when  ${\cal H}_F \simeq \mathbb{C}$; see \cite[Definition 8 and subsection II-A]{TV08} for details.
\end{remark}

\subsection{Physical interpretation}\label{subsec:physical}
%?????

The discussions of the Tavis-Cummings model in the previous subsection  in terms of linear quantum systems theory  appear purely mathematical;  however, these results can indeed reveal several typical features of the Tavis-Cummings model.

Firstly, from \eqref{eq:transfer_3}, it can be seen that the collection of the $N$ atoms act as a giant atom of detuned frequency $\omega_s$ and coupling strength $\sqrt{N}\bar{\Gamma}$, which reflects the $\sqrt{N}$- scaling of the collective coupling strength; in other words, this giant atom decays $N$ times as fast as a single atom. This is the physical basis of  superradiance.

Secondly, the two peaks of the transfer function $|T[\mathrm{i}\omega]|^2$ in \eqref{eq:T_s2} at $\omega=\pm \sqrt{N}\bar{\Gamma}$ echo the vacuum Rabi mode splitting in the Tavis-Cummings model, see, e.g., \cite[Fig. 2(b)]{MCG07}, \cite[Fig. 3]{Fink09}, and \cite[Fig. 2(b)]{WW+PRL20}.

Thirdly, the controllable and observable subsystem echos the bright states and the uncontrollable and unobservable subsystems echo the dark states of  the Tavis-Cummings model. For simplicity, we look at the simplest case of $\omega_1=\cdots=\omega_N\equiv\omega_s$. In this case, according to Corollary \ref{cor:kalman}, there is an orthogonal matrix   $T$ which yields a system with system matrices given in \eqref{eq:jun4_ABC_4}.
%
%
%
%{\color{blue}\begin{remark}
%Let $N=3$, and $\omega_1=\omega_2=\omega_3\equiv\omega_s$, i.e., three atoms are coupled with the cavity and resonant with each other. The system matrix $A$ can be factorized with a double root $\mathrm{i}\omega_s$. In this case, the coupling constants are choosen to be $g_1=-g_2=g_3$ \cite[Fig. 3(f)]{Fink09}, and the corresponding eigenstate $x=[x_1~x_2~x_3~x_4]^\top$ of eivenvalue $\mathrm{i}\omega_s$ can be calculated by
%\begin{equation}\label{nov7-1}\begin{aligned}
%Ax=\mathrm{i}\omega_sx,
%\end{aligned}\end{equation}
%which yields $x_1-x_2+x_3=0$ and $x_4=0$. We choose two linearly independent eigenstates $[1~0~-1~0]^\top$ and $[0~1~1~0]^\top$. Then the system states are given by
%\begin{equation}\label{nov7-2}\begin{aligned}
%&|e_1g_2g_30\rangle-|g_1g_2e_30\rangle, \\
%&|g_1e_2g_30\rangle+|g_1g_2e_30\rangle,
%\end{aligned}\end{equation}
%which are two dark states under the eigenvalue $\mathrm{i}\omega_s$ and as same as $|3,1d_1\rangle$, $|3,1d_2\rangle$ given in \cite{Fink09}. Consequently, the $N-1$ eigenvalues $\mathrm{i}\omega_s$ of system matrix $A$ are corresponded to $N-1$ dark states of the Tavis-Cummings model; while the remaining two eigenvalues $\lambda_{sN}$, $\lambda_{s(N+1)}$ are corresponded to two bright states of the Tavis-Cummings model.
%\end{remark}
%}
Let $N=3$ and assume the coupling constants $\Gamma_1=-\Gamma_2=\Gamma_3$. (In \cite{Fink09}, $\Gamma_1,\Gamma_2,\Gamma_3$ are respectively $g_\mathrm{A},g_\mathrm{B},g_\mathrm{C}$. As shown in \cite[Table 1]{Fink09}, the actual values of $\Gamma_1,-\Gamma_2,\Gamma_3$ are not exactly identical, but the discrepancy has negligible effect as can be seen from consistency between the red region (for real data) and the white dashed curves (for theoretical calculation) in \cite[Fig. 3]{Fink09}.) The orthogonal transformation matrix $T$ can be calculated as $T_1=\frac{1}{\sqrt{2}}\left[
                         \begin{array}{cccc}
                           1 & 1 & 0 & 0 \\
                         \end{array}
                       \right]^\top$, $T_2=\frac{1}{\sqrt{6}}\left[
                         \begin{array}{cccc}
                           -1 & 1 & 2 & 0 \\
                         \end{array}
                       \right]^\top$, $T_3=\frac{1}{\sqrt{3}}\left[
                         \begin{array}{cccc}
                           1 & -1 & 1 & 0 \\
                         \end{array}
                       \right]^\top$, and   $T_4=\left[
                         \begin{array}{cccc}
                           0 & 0 & 0 & 1 \\
                         \end{array}
                       \right]^\top$.
%\begin{equation*}\label{May31-3}\begin{aligned}
%&T_1=\frac{1}{\sqrt{2}}\left[
%                         \begin{array}{cccc}
%                           1 & 1 & 0 & 0 \\
%                         \end{array}
%                       \right]^\top,  \
%T_2=\frac{1}{\sqrt{6}}\left[
%                         \begin{array}{cccc}
%                           -1 & 1 & 2 & 0 \\
%                         \end{array}
%                       \right]^\top, \\
%&T_3=\frac{1}{\sqrt{3}}\left[
%                         \begin{array}{cccc}
%                           1 & -1 & 1 & 0 \\
%                         \end{array}
%                       \right]^\top,  \
%                       T_4=\left[
%                         \begin{array}{cccc}
%                           0 & 0 & 0 & 1 \\
%                         \end{array}
%                       \right]^\top.
%\end{aligned}\end{equation*}
Identify $g$ with $0$ and $e$ with $1$ respectively. It can be verified that the dark states $|3,1d_1\rangle = \frac{1}{\sqrt{2}} (\ket{e,g,g,0}-\ket{g,g,e,0})$ and $|3,1d_2\rangle =  \frac{1}{\sqrt{2}} (\ket{g,e,g,0}+\ket{g,g,e,0})$ in \cite{Fink09} can be expressed as $|3,1d_1\rangle=\frac{1}{2}T_1-\frac{\sqrt{3}}{2}T_2$, and  $|3,1d_2\rangle=\frac{1}{2}T_1+\frac{\sqrt{3}}{2}T_2$ respectively,
%\begin{equation*}\begin{aligned}
%|3,1d_1\rangle=\frac{1}{2}T_1-\frac{\sqrt{3}}{2}T_2, ~ |3,1d_2\rangle=\frac{1}{2}T_1+\frac{\sqrt{3}}{2}T_2;
%\end{aligned}\end{equation*}
while the bright states $|3,1\pm\rangle = \frac{1}{\sqrt{2}}\ket{g,g,g,1}\pm \frac{1}{\sqrt{6}}\ket{e,g,g,0}$ in \cite{Fink09} can be written as $|3,1\pm\rangle=\frac{1}{\sqrt{2}}(T_4\pm T_3)$.
%\begin{equation*}
%|3,1\pm\rangle=\frac{1}{\sqrt{2}}(T_4\pm T_3).
%\end{equation*}
In other words, the dark states $|3,1d_1\rangle$ and $|3,1d_2\rangle$ live in the decoherence-free subspace spanned by $T_1$ and $T_2$, while the bright states $|3,1\pm\rangle$ live in the controllable and observable subspace spanned by $T_3$ and $T_4$.

Take the case $N=2$ for another example. Similar correspondence can be found between the eigenstates of the Tavis-Cummings model \eqref{system1} and the vectors in \eqref{feb22-1}. Regard the states $|g,e,0\rangle$, $|e,g,0\rangle$, $|g,g,1\rangle$ in \cite{Van18} as vectors $[0~1~0]^\top$, $[1~0~0]^\top$, and $[0~0~1]^\top$, respectively. In the resonant ($\omega_r=\omega_s$) case, the dark state $|0\rangle_{r3} = \frac{1}{\sqrt{N}\bar{\Gamma}}(\Gamma_1\ket{g,e,0}-\Gamma_2\ket{e,g,0})$ in \cite[Appendix C-3]{Van18} is $-T_1$, and the two bright states $|\pm\rangle_{r3} =\frac{1}{N\bar{\Gamma}} (\Gamma_2\ket{g,e,0}+\Gamma_1\ket{e,g,0}\pm \sqrt{N}\bar{\Gamma}\ket{g,g,1})$ in \cite[Appendix C-3]{Van18} are  $\frac{1}{\sqrt{2}}(T_2\pm T_3)$. Also, $|0\rangle_{r3}$ and $|\pm\rangle_{r3}$ are respectively $-\ket{2,1d}$ and $-\ket{2,1\pm}$ in \cite{Fink09} when $\Gamma_1=-\Gamma_2=\Gamma_3$. Moreover, in the dispersive regime ($\Delta_r=|\omega_r-\omega_s|\gg \Gamma_j$, $j=1,\ldots,N$) considered in \cite[Appendix C]{Van18}, the three states $|+^\prime\rangle_{r3}, |-^\prime\rangle_{r3}, |1^\prime\rangle_{r3}$ given in \cite[(C18)]{Van18} can be expressed by
 $|+^\prime\rangle_{r3}=-T_1$, $|-^\prime\rangle_{r3}=\frac{1}{\sqrt{N\bar{\Gamma}^2+\Delta_r^2}}(\Delta_r T_2-\sqrt{N}\bar{\Gamma}T_3)\approx T_2$, and $|1^\prime\rangle_{r3}=\frac{1}{\sqrt{N\bar{\Gamma}^2+\Delta_r^2}}(\sqrt{N}\bar{\Gamma}T_2 + \Delta_r T_3) \approx T_3$.
%\begin{equation*}\label{May29-2}\begin{aligned}
%&|+^\prime\rangle_{r3}=-T_1, \\
%&|-^\prime\rangle_{r3}=\frac{1}{\sqrt{N\bar{g}^2+\Delta_r^2}}(\Delta_r T_2-\sqrt{N}\bar{g}T_3)\approx T_2, \\
%&|1^\prime\rangle_{r3}=\frac{1}{\sqrt{N\bar{g}^2+\Delta_r^2}}(\sqrt{N}\bar{g}T_2 + \Delta_r T_3) \approx T_3.
%\end{aligned}\end{equation*}
Thus, in both cases, the dark  states live in the space spanned by $T_1$ while the bright states live in the space  spanned by $T_2$ and $T_3$.

Fourthly, the dynamics of two remote spin ensembles coupled by a cavity bus are experimentally studied in \cite{ANP+17}. The Hamiltonian of the system is given in \cite[(2)]{ANP+17}, which is of the form of $H_{\mathrm{TC}}$ for the Tavis-Cummings model.
When $\varphi\approx 48.1^\mathrm{o}$ in  \cite[Fig. 2(b)]{ANP+17}, both spin ensembles resonate with the cavity mode.  Thus, this particularly interesting case  can be analyzed by Corollary \ref{cor:kalman}. Indeed,  the eigenstates $\ket{\pm}$ in \cite[(3)]{ANP+17} and the dark mode $\ket{D}$ in \cite[(4)]{ANP+17} can be obtained by means of the orthonormal matrix $T$ in  Corollary \ref{cor:kalman}.

Fifthly, in subsections \ref{sec:output} and \ref{sec:outputsimu} we study how the Tavis-Cummings model \eqref{system1} responses to a continuous-mode single-photon input state, where the linear model developed in subsection \ref{subsec:linear} plays an essential role.

Finally, the single-excitation superradiant and subradiant states of the Tavis-Cummings model can be analyzed by means of the quantum linear systems theory presented in subsection \ref{subsec:linear}; see Remark \ref{rem: sup_sub} in subsection \ref{sec:super}.

%%%%%%%%%%%%%%%%%%%%%%%%%%
%%%%%%%%%%%%%%%%%%%%%%%%%%
%%%%%%%%%%%%%%%%%%%%%%%%%%
\section{The single-excitation case}\label{sec:single_excitation}
In this section, we investigate the dynamics of the Tavis-Cummings model when there is only one excitation. %Specifically, in subsections \ref{sec:output} and \ref{sec:outputsimu}  we study how the Tavis-Cummings model responses to a continuous-mode single-photon input state. In subsection \ref{sec:super} we derive the analytic form of the joint system-field state when initially only one of the two-level atom is in the excited state, the cavity is empty and the incident field is in the vacuum state.

%%%%%%%%%%%%%%%%%%%%%%%%%%
%%%%%%%%%%%%%%%%%%%%%%%%%%
\subsection{Response to single-photon inputs}\label{sec:output}

In this subsection, we derive an analytic expression of the steady-state output field state when the  Tavis-Cummings model is initialized in the state $\ket{\zeta}$  given in \eqref{feb3:initial}, and driven by a single-photon state.

We start with the following lemma which discusses the controllability (\cite[Sec. III-B]{NY14tac}, \cite[Definition 1]{GZ15}) of the passive linear  quantum system \eqref{linear_system}.

%%%%%%%%%%%%%%%%%%%%%%%%%%%%%
\begin{lemma}\label{stability}
The passive linear  quantum system \eqref{linear_system} is controllable if and only if $\omega_j\neq \omega_k$ for all $1\leq j<k\leq N$.
\end{lemma}

%%%%%%%%%%%%%%%%%%%%%%%%%%%%%
\begin{proof}
Let a scalar $\lambda\in\mathbb{C}$ and a vector $x=[x_1,\ldots,x_{N+1}]^\top\in\mathbb{C}^{N+1}$ satisfy
\begin{equation}\label{eq:feb16_1a}
Ax=\lambda x, \ \ x^\dagger B=0.
\end{equation}
We have
\begin{equation}\label{eq:feb16_2a}
x_{N+1}=0, \ \ \ -\mathrm{i} \omega_k x_k = \lambda x_k, \ \ \ k=1,\ldots, N,
\end{equation}
and
\begin{equation}\label{eq:feb16_2c}
\sum_{k=1}^N \Gamma_k x_k =0.
\end{equation}

\emph{Necessity}. We prove it by contradiction. Without loss of generality, assume $\omega_1=\omega_2$. We choose $\lambda=-\mathrm{i}\omega_1$,  $x_1=\Gamma_2$, $x_2=-\Gamma_1$, and $x_j=0$ for all $j=3,\ldots, N+1$. Clearly, the scalar $\lambda$ and the {\it nonzero} vector $x$ satisfy \eqref{eq:feb16_1a}. This means that the system is not controllable. A contradiction is reached.

\emph{Sufficiency}. If $x_j\neq 0$ for some $j=1,\dots, N$, then by \eqref{eq:feb16_2a} $\lambda=-\mathrm{i}\omega_j$. As $\omega_j\neq\omega_k$ for all $1\leq j<k\leq N$, again by \eqref{eq:feb16_2a} $x_k=0$ for all $j\neq k$. Thus, \eqref{eq:feb16_2c} reduces to $\Gamma_jx_j=0$, which yields $x_j=0$. This shows that any vector $x$ satisfying \eqref{eq:feb16_1a} must be a zero vector. Hence, the system is controllable.
\end{proof}

For system \eqref{system_linear}, $\tilde{\omega}_j \neq \tilde{\omega}_k$ for all $1\leq j<k\leq M$. By Lemma \ref{stability}, system \eqref{system_linear} is Hurwitz stable.

With the aid of Lemma \ref{stability}, the main result of this section can be derived.

%%%%%%%%%%%%%%%%%%%%%%%%%%
\begin{theorem}\label{theoutput}
Assume that the Tavis-Cummings model \eqref{system1} is initialized in the state $\ket{\zeta}=|g_1g_2\cdots g_N\rangle\otimes|0\rangle$ and driven by a single-photon input state with pulse shape $\xi$. The steady-state ($t\to \infty$ and $t_0\to -\infty$) output field state is a single-photon state with the frequency-domain pulse shape
\begin{equation}\label{outputspe}
\eta[\mathrm{i}\omega]=G[\mathrm{i}\omega]\xi[\mathrm{i}\omega],
\end{equation}
where the transfer function $G[s]$ is given by \eqref{transfer}.
\end{theorem}

%%%%%%%%%%%%%%%%%%%%%%%%%%%
\begin{proof}
By system \eqref{system1} and \eqref{eq:z_g}, we have
%\[
%\bra{\zeta \Phi_0} X(t) = e^{A(t-t_0)} \bra{\zeta \Phi_0} X(t_0) + \int_{t_0}^t e^{A(t-\tau)}  \bra{\zeta \Phi_0} b(\tau) d\tau.
%\]
%Hence,
\beqnm
\bra{\zeta \Phi_0} b_{\rm out}(t) &=& Ce^{A(t-t_0)} \bra{\zeta \Phi_0} X(t_0) \\
&&+ \int_{t_0}^t Ce^{A(t-\tau)}  \bra{\zeta \Phi_0} b_{\rm in}(\tau) d\tau + \bra{\zeta \Phi_0} b(t).
\eeqnm
If $A$ is Hurwitz stable, then $Ce^{A(t-t_0)} \bra{\zeta \Phi_0} X(t_0)  \to 0$ as $t_0\to -\infty$.  If $A$ is not Hurwitz stable, then by Corollary \ref{cor:kalman}, the $\bar{c}\bar{o}$ subsystem does not affect the input-output behavior, while the $co$ subsystem is Hurwitz stable; see Proposition \ref{prop:3nov} in the Appendix. Hence we also have $Ce^{A(t-t_0)} \bra{\zeta \Phi_0} X(t_0)  \to 0$ as $t_0\to -\infty$. As a result, sending $t_0\to -\infty$, we get
\beq
\bra{\zeta \Phi_0} b_{\rm out}(t)  = \int_{-\infty}^\infty g_G(t-r) \bra{\zeta \Phi_0} b_{\rm in}(r) dr,
\eeq
where $g_G(t)$ is the impulse response function associated to the transfer function $G[s]$ in \eqref{tran}.
As $b_{\rm out}(t) = U^\ast b(t,t_0)_{\rm in}(t) U(t,t_0)$, we have
\beq
\bra{\zeta \Phi_0} b_{\rm in}(t)  = \int_{-\infty}^\infty g_G(t-r) \bra{\zeta \Phi_0} b^-(r,-\infty) dr,
\eeq
where $b^-(t,t_0) \triangleq U(t,t_0) b_{\rm in}(t) U^\ast(t,t_0)$. Then following the stable inverse technique in \cite[Lemma 1]{ZJ13}, we can get
\beqn
 b^{-\ast}(r,-\infty) \ket{\zeta \Phi_0}
 &=&\int_{-\infty}^\infty g_{G^{-1}}^\ast(r-t) b_{\rm in}^\ast(t) dt \ket{\zeta \Phi_0} \nonumber \\
 &=&  \int_{-\infty}^\infty g_G(t-r) b_{\rm in}^\ast(t) dt \ket{\zeta \Phi_0},
\eeqn
and thus
\beq
  \int_{-\infty}^\infty \xi(r) b^{-\ast}(r,-\infty) \ket{\zeta \Phi_0}  dr=  \int_{-\infty}^\infty \eta(t) b_{\rm in}^\ast(t)  \ket{\zeta \Phi_0} dt ,
\eeq
where $\eta(t)$ is the time-domain counterpart of $\eta[\mathrm{i}\omega]$ in \eqref{outputspe}. Consequently, the steady-state joint system-field state in \cite[Eq. (85)]{ZJ13} is
\beq
\rho_\infty = \ket{\zeta}\bra{\zeta}\otimes \ket{1_\eta}\bra{1_\eta}.
\eeq The steady-state output field state  $\rho_{\rm out}$ is then obtained by tracing over the initial system state, i.e.,  $\rho_{\rm out}=\mathrm{Tr}_{\rm sys}[\rho_\infty ] =\ket{1_\eta}\bra{1_\eta} $, which is a pure state $\ket{1_\eta}$.
\end{proof}

\subsection{Simulation results for $N\leq4$}\label{sec:outputsimu}
In this subsection, we illustrate Theorem \ref{theoutput} by a special case of $N\leq4$, i.e., at most four two-level atoms are coupled to the cavity. The single-photon input state is supposed to have a rising exponential  pulse shape
\begin{equation}\label{eq:xi}
\xi(t)=\left\{\begin{array}{cc}
                \sqrt{\gamma}e^{\frac{\gamma}{2}t}, & t\leq0, \\
                0, & t>0,
              \end{array}
\right.
\end{equation}
where $\gamma$ denotes the full width at half maximum (FWHM) of the Lorentzian spectrum. The input and output photon probability distributions are shown in Fig. \ref{distri1}. It is worthwhile to notice that the carrier frequency of the single-photon input field is not shown in  \eqref{eq:xi}, the reason is that all the frequencies in the system Hamiltonian $H_{\mathrm{TC}}$ in  \eqref{H_TC} are detuned from this carrier frequency.

\begin{figure}[htp!]
\centering
\includegraphics[scale=0.6]{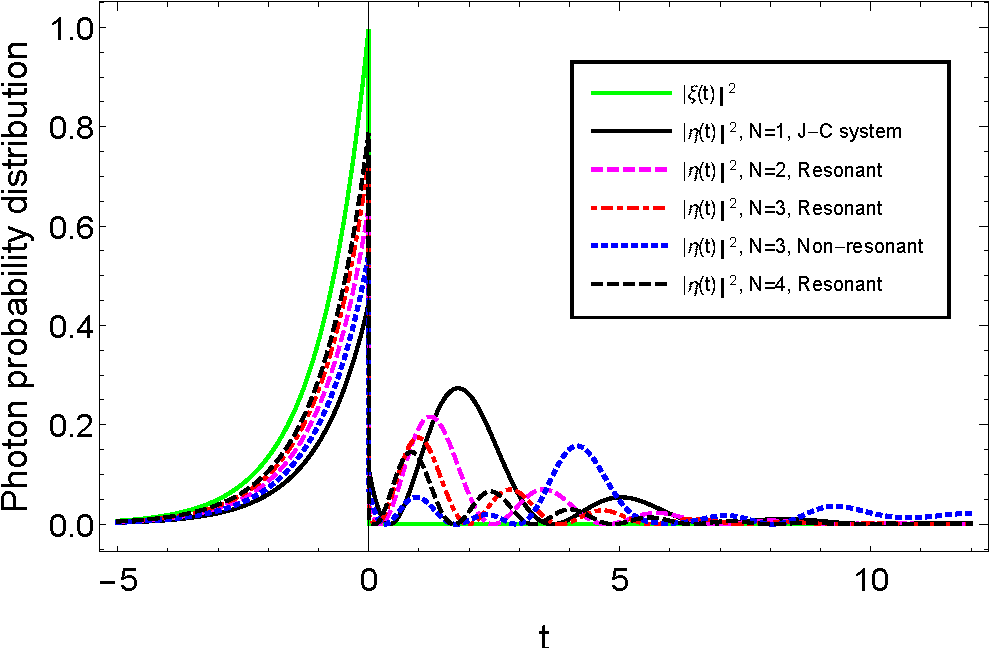}
\caption{The input and output photon probability distributions. The input photon probability distribution is plotted as the green curve; the red dot dashed curve plots the output photon probability distribution in the case that all the three two-level atoms are resonant with each other ($\omega_1=\omega_2=\omega_3=0$); the non-resonant case is plotted as the blue dotted curve ($\omega_1=1$, $\omega_2=-1$, $\omega_3=0$); the output photon probability distributions for the resonant cases of $N=1$ (Jaynes-Cummings system), $N=2$ and $N=4$ are shown as the black solid, magenta dashed, and black dashed curves respectively. $(\gamma=\kappa=1$, $\omega_r=0$, and $\Gamma_j=1$, $j=1,\ldots,4.$)}
\label{distri1}
\end{figure}

In Fig. \ref{distri1}, it can be observed that the emitted photon is more likely to be found when $t<0$ as the number of atoms increases, which means it interacts with the Jaynes-Cummings system (the $N=1$ case) more easily than with the Tavis-Cummings model. On the other hand, the Rabi oscillation at $t>0$  indicates that the photon can be repeatedly absorbed and emitted by the two-level atoms. The oscillation also becomes stronger when there are more atoms, which indicates that the added atoms increase the time for the photon to escape from the cavity. The Rabi oscillation monotonically decays when all the atoms are resonant with each other, while revivals can be observed in the non-resonant case (the blue dotted curve). Moreover, oscillation sustains much longer in the non-resonant case than in the resonant case. Finally, simulation shows that in the resonant case, the shapes of the input and output pulses are quite close to each other when $N$ is large, showing that the input single photon hardly interacts with the atomic ensemble. This phenomenon is confirmed by $\lim_{N\to \infty}G[\i \omega]=1$ for any fixed  $\omega$, where $G[s]$ is given in  \eqref{eq:transfer_3}.

%%%%%%%%%%%%%%%%%%%%%%%%%%
%%%%%%%%%%%%%%%%%%%%%%%%%%
\subsection{An analytic form of the superposition state}\label{sec:super}

%When the two-level atoms are resonant with each other and equally coupled to the cavity,  and one of them is initialized in the excited state, the blue and red solid curves in Figs. \ref{single+excited} and \ref{two_excited_new} of the last section  demonstrate that the final excitation probability of the first two-level atom is non-zero.

In this subsection, an analytic form of the joint system-field state is derived.

%which explains the above observations. More analysis of the other simulations in Figs. \ref{single+excited} and \ref{two_excited_new} will be given in Section \ref{sec:alter}.

The following theorem is the main result of this subsection.

%%%%%%%%%%%%%%%%%%%%%%%%%%
\begin{theorem}\label{state}
Assume that the $N$ two-level atoms of the Tavis-Cummings model are resonant with each other, i.e., $\omega_1=\cdots=\omega_N\equiv\omega_s$, the $k$th two-level atom  is in the excited state, the others are in the ground state, the cavity is empty and the Tavis-Cummings model is driven by the vacuum input state. That is, the initial joint system-field state is $|\Psi_k(0)\rangle=\ket{g_1g_2\ldots e_k\ldots g_N0}\otimes\ket{\Phi_0}$. Then the joint system-field state is
\begin{align}
\ket{\Psi_k(t)}=&\ c_k(t)|g_1g_2\cdots e_k\cdots g_N0\Phi_0\rangle \nonumber \\
&+\sum_{j\neq k}^{N}c_j(t)|g_1g_2\cdots e_j\cdots g_N0\Phi_0\rangle  \nonumber \\
&+c_{N+1,k}(t)\int_{0}^{t}\varphi(\tau)dB_{\mathrm{in}}^\ast(\tau)|g_1g_2\cdots g_N0\Phi_0\rangle  \nonumber \\
&+c_{N+2,k}(t)|g_1g_2\cdots g_N1\Phi_0\rangle,
\label{superposition}
\end{align}
where
\begin{align}
&c_k(t)=\frac{e^{\frac{N-2}{2}\mathrm{i}\omega_s t}}{2\chi N \bar{\Gamma}^2}\Bigg[2\chi\sum_{j\neq k}^{N}\Gamma_j^2
+\Gamma_k^2\Big(\lambda_1 e^{-\frac{\lambda_2 t}{4}} -\lambda_2 e^{-\frac{\lambda_1 t}{4}}\Big)\Bigg],  \nonumber \\
&c_j(t)=\frac{-\Gamma_j\Gamma_k e^{\frac{N-2}{2}\mathrm{i}\omega_st}}{2\chi N \bar{\Gamma}^2}\Bigg[2\chi
-\left(\lambda_1 e^{-\frac{\lambda_2 t}{4}} -\lambda_2 e^{-\frac{\lambda_1 t}{4}}\right)\Bigg], j\neq k,  \nonumber \\
&c_{N+1,k}(t)=\frac{\Gamma_k}{\sqrt{N}\bar{\Gamma}}e^{\frac{N-2}{2}\mathrm{i}\omega_st},  \nonumber \\
&c_{N+2,k}(t)=\frac{-4\mathrm{i}\Gamma_k\sinh\left(\frac{1}{4}\chi t\right)}{\chi}
e^{-\frac{1}{4}[\kappa+2\mathrm{i}(\omega_r-(N-1)\omega_s)]t}, \nonumber  \\
&\varphi(\tau)=\frac{-4\mathrm{i}\bar{\Gamma}\sqrt{\kappa N}\sinh\left(\frac{1}{4}\chi \tau\right)}{\chi}e^{-\frac{\kappa \tau + 2\i [\omega_s (\tau-t)+\omega_r \tau - \omega_s t ]}{4}},  \nonumber \\
& \hspace{6ex}  (0\leq \tau\leq t), \label{oct27-1}
\end{align}
with $\lambda_1 = \kappa+\chi+2\mathrm{i}(\omega_r-\omega_s)$, $\lambda_2 = \kappa-\chi+2\mathrm{i}(\omega_r-\omega_s)$, and $\chi=\sqrt{(\kappa+2\mathrm{i}\omega_r-2\mathrm{i}\omega_s)^2-16N \bar{\Gamma}^2}$.
%\begin{equation*}\label{oct27-2}\begin{aligned}
%&\lambda_1 = \kappa+\chi+2\mathrm{i}(\omega_r-\omega_s), ~ \lambda_2 = \kappa-\chi+2\mathrm{i}(\omega_r-\omega_s),\\
%&\chi=\sqrt{(\kappa+2\mathrm{i}\omega_r-2\mathrm{i}\omega_s)^2-16N \bar{g}^2}.
%\end{aligned}\end{equation*}
Moreover, in the steady state ($t = \infty$), by ignoring the rotating term $e^{\frac{N-2}{2}\mathrm{i}\omega_st}$ in the coefficients in \eqref{oct27-1}  (equivalently, setting $\omega_s = 0$) the superposition state \eqref{superposition} becomes
\begin{equation}\label{oct27-3}\begin{aligned}
|\Psi_k(\infty)\rangle=&\ c_k(\infty)|g_1g_2\cdots e_k\cdots g_N0\Phi_0\rangle \\
&+\sum_{j\neq k}^{N}c_j(\infty)|g_1g_2\cdots e_j\cdots g_N0\Phi_0\rangle \\
&+c_{N+1,k}(\infty)|g_1g_2\cdots g_N01_{\varphi}\rangle,
\end{aligned}\end{equation}
where
\begin{equation}\label{oct27-4}\begin{aligned}
&c_k(\infty)=\frac{1}{N \bar{\Gamma}^2} \sum_{j\neq k}^{N}\Gamma_j^2,  \ \ c_j(\infty)=\frac{-\Gamma_j \Gamma_k}{N \bar{\Gamma}^2}, ~ j\neq k, \\
&c_{N+1,k}(\infty)=\frac{\Gamma_k}{\sqrt{N} \bar{\Gamma}}.
\end{aligned}\end{equation}
\end{theorem}

%%%%%%%%%%%%%%%%%%%%%%%%%%
The proof of Theorem \ref{state} follows the recursive formula for joint system-field states of  general quantum systems; see Remark \ref{rem:for thm 5.1} in Section \ref{sec:alter}.
\begin{remark}
By Theorem \ref{state}, it can be seen that  the cavity is eventually empty ($c_{N+2,k}(\infty)=0$), which results in a superposition state of the two-level atoms and the output field. Moreover, by  \eqref{oct27-4}, the steady-state excitation probabilities of two-level atoms are independent of the cavity resonant frequency $\omega_r$ as well as the  atomic transition frequency $\omega_s$. %Instead, they can be designed by tuning the coupling strengths $g_j$ between the two-level atoms and the cavity.
\end{remark}

It is well-known in quantum optics  that the collective radiation of an ensemble of two-level atoms can be accelerated by superradiance or inhibited by subradiance \cite{Dicke54,TC68,WW+PRL20}, the following corollary gives the steady state of the Tavis-Cummings model initialized in a superposition state of the superradiant and subradiant states.

%%%%%%%%%%%%%%%%%%%%%%%%%%
\begin{corollary}\label{cor1}
Assume that the $N$ two-level atoms in the Tavis-Cummings model are resonant with each other, i.e., $\omega_1=\cdots=\omega_N\equiv\omega_s$.  Let  the  system be driven by the vacuum input state and initialized in the following superposition state $(\alpha\ket{B_N}+\beta\ket{D_N})\otimes \ket{0}$,
%\begin{equation*}\label{oct22-1}\begin{aligned}
%\alpha|B_N0\rangle+\beta|D_N0\rangle,
%\end{aligned}\end{equation*}
where $|\alpha|^2+|\beta|^2=1$, and the single-excitation superradiant state $\ket{B_N}$ and subradiant state $\ket{D_N}$ are respectively
\begin{equation}\label{oct22-2}\begin{aligned}
&\ket{B_N}=\frac{1}{\sqrt{N}}\sum_{k=1}^{N}\ket{g_1\cdots e_k\cdots g_N}, \\
&\ket{D_N}=\frac{1}{\sqrt{N}}\sum_{k=1}^{N}e^{-\mathrm{i}\phi_k}
\ket{g_1\cdots e_k\cdots g_N} \\
\end{aligned}\end{equation}
with $\phi_k=\frac{2\pi}{N}k$. Then the joint system-field state is
\begin{equation}\label{oct22-3}
|\Psi^\prime(t)\rangle=\frac{1}{\sqrt{N}}\sum_{k=1}^{N}\left(\alpha+\beta e^{-\mathrm{i}\phi_k}\right)|\Psi_k(t)\rangle,
\end{equation}
where $|\Psi_k(t)\rangle$ is given in \eqref{superposition}. Moreover, in the steady state ($t=\infty$), by ignoring the rotating term $e^{\frac{N-2}{2}\mathrm{i}\omega_st}$ in the coefficients in \eqref{oct27-1}  (equivalently, setting $\omega_s = 0$), the superposition state \eqref{oct22-3} becomes
\begin{equation}\label{feb2-1}
|\Psi^\prime(\infty)\rangle=\frac{1}{\sqrt{N}}\sum_{k=1}^{N}\left(\alpha+\beta e^{-\mathrm{i}\phi_k}\right)|\Psi_k(\infty)\rangle,
\end{equation}
where $|\Psi_k(\infty)\rangle$ is given in \eqref{oct27-3}.
\end{corollary}

The proof of Corollary \ref{cor1} is omitted.
%Moreover, the superradiant state $\ket{B_N}$ and subradiant state $\ket{D_N}$ are orthogonal due to
%\begin{equation}
%\langle B_N\ket{D_N}=\frac{1}{N}\sum_{k=1}^{N}e^{-\mathrm{i}\phi_k}=0.
%\end{equation}

Let $\alpha=0$, $\beta=1$, i.e., the Tavis-Cummings model is initialized in the pure state $|D_N0\rangle$. By Corollary \ref{cor1}, the joint system-field steady state is
\begin{equation}\label{oct24-2}
|\Psi^\prime(\infty)\rangle=\frac{1}{\sqrt{N}}\sum_{k=1}^{N}c_k^\prime(\infty)|g_1\cdots e_k\cdots g_N0\Phi_0\rangle,
\end{equation}
where
%\begin{equation}
%\red{c_k^\prime(\infty)=\frac{e^{-\mathrm{i}\phi_k}\sum_{j\neq k}^{N}g_j^2-g_k\sum_{j\neq k}^{N}e^{-\mathrm{i}\phi_j}g_j}{N\bar{g}^2}e^{-\frac{N-2}{2}\mathrm{i}\omega_st}.}
%\end{equation}
\begin{equation}
c_k^\prime(\infty)=\frac{e^{-\mathrm{i}\phi_k}\sum_{j\neq k}^{N}\Gamma_j^2-\Gamma_k\sum_{j\neq k}^{N}e^{-\mathrm{i}\phi_j}\Gamma_j}{N\bar{\Gamma}^2},
\end{equation}
provided that $\omega_s=0$.
Assume further that the atoms are all equally coupled with the cavity, i.e., $\Gamma_1=\Gamma_2=\cdots=\Gamma_N$, then the steady state \eqref{oct24-2} reduces to
\begin{equation}\label{feb22-4}
|\Psi^\prime(\infty)\rangle=|D_N0\rangle\otimes\ket{\Phi_0}.
\end{equation}
In other words, in the steady state, the single excitation only exists in the two-level atoms which have the same excitation probability $\frac{1}{N}$. The Tavis-Cummings system cannot emit a photon into the cavity or  the output field. This theoretical result is consistent with the experimental results given in \cite[Fig. 2(c), Fig. 4(a)]{WW+PRL20}, where small fluctuations of the collective swapping dynamics are due to the inhomogeneity of the coupling strengths. In fact, when $\Gamma_1=\Gamma_2=\cdots=\Gamma_N$ and $\omega_s=0$, by Corollary \ref{cor:kalman}, $\hat{A}_{\bar{c}\bar{o}}=0$, and thus the subsystem $(\hat{A}_{\bar{c}\bar{o}}, \hat{B}_{\bar{c}\bar{o}},\hat{C}_{\bar{c}\bar{o}})$ is static. On the other hand, it is easy to show that $\ket{\Psi(t)}\equiv \ket{\Psi(0)} = |D_N0\rangle\otimes\ket{\Phi_0}$ for all $t\geq0$.

Let $\alpha=1$, $\beta=0$, i.e., the Tavis-Cummings model is initialized in the pure state $|B_N0\rangle$. In this case, the steady state is
\begin{equation}\label{oct24-3}\begin{aligned}
|\Psi^\prime(\infty)\rangle=&\frac{1}{\sqrt{N}}\sum_{k=1}^{N}\bigg(c^\prime_k(\infty)|g_1\cdots e_k \cdots g_N0\Phi_0\rangle \\
&+c_{N+1,k}(\infty)|g_1g_2 \cdots g_N01_{\varphi}\rangle\bigg),
\end{aligned}\end{equation}
where the pulse shape $\varphi$ is given in \eqref{oct27-1}, and $c^\prime_k(\infty)=\frac{\sum_{j\neq k}^{N}\Gamma_j^2-\Gamma_k\sum_{j\neq k}^{N}\Gamma_j}{N\bar{\Gamma}^2}$,
%\begin{equation*}\label{oct24-4}\begin{aligned}
%c^\prime_k(\infty)=\frac{\sum_{j\neq k}^{N}g_j^2-g_k\sum_{j\neq k}^{N}g_j}{N\bar{g}^2},
%\end{aligned}\end{equation*}
provided that $\omega_s=0$,
and $c_{N+1,k}(\infty)$ is given in \eqref{oct27-4}. Assume further that the atoms are all equally coupled with the cavity, i.e., $\Gamma_1=\Gamma_2=\cdots=\Gamma_N$, the steady state \eqref{oct24-3} reduces to
\begin{equation} \label{eq:superradiant}
|\Psi^\prime(\infty)\rangle=|g_1g_2 \cdots g_N01_{\varphi}\rangle.
\end{equation}
In other words, due to the existence of cavity decay rate $\kappa$, the Tavis-Cummings system eventually emits a photon into the field.

%%%%%%%%%%%%%%%%%%%%%%%%%%
\begin{remark}\label{rem: sup_sub}
The superradiant state $\ket{B_N}$ and subradiant state $\ket{D_N}$ can be represented by the eigenvectors given in \eqref{feb22-1}, respectively. For simplicity, we assume the coupling strengths $\Gamma_1=\cdots=\Gamma_N\equiv \Gamma$. If the state $|g_1\cdots e_k \cdots g_N0\rangle$ is viewed as a column vector $[0\cdots 1 \cdots 0]^\top$ (the $k$th element is $1$ and the others are $0$), then $|B_N0\rangle$ is exactly the vector $T_N$ in \eqref{feb22-1}. Moreover, $A|B_N0\rangle=-\mathrm{i}\omega_sT_N-\mathrm{i}\sqrt{N}\Gamma T_{N+1}$,
%\begin{equation*}
%A|B_N0\rangle=\mathrm{i}\omega_sT_N-\mathrm{i}\sqrt{N}gT_{N+1},
%\end{equation*}
which lives in the linear span of the two vectors $T_{N}$ and $T_{N+1}$ of the controllable/observable subsystem $\hat{A}_{co}$ given in \eqref{feb22-2}.  Hence, the single excitation swaps among all the two-level atoms and the cavity. This echos the collective swapping dynamics shown in \cite[Fig. 2(a)]{WW+PRL20}. Moreover, as shown in  \eqref{eq:superradiant}, eventually all the oscillations will die out  as the photon is emitted into the external field due to the lossy nature of the cavity. On the other hand, the subradiant state $|D_N0\rangle$ can be represented by
\begin{equation}\label{feb23-1}
|D_N0\rangle=\sum_{j=1}^{N-1}\alpha_jT_j,
\end{equation}
where $\alpha_j=-\sqrt{\ff{j+1}{jN}}\left[e^{-\mathrm{i}\phi_{j+1}}+\frac{1}{j+1}\sum_{k=j+2}^{N}e^{-\mathrm{i}\phi_k}\right]$,
%\begin{equation*}
%\alpha_j=-\frac{\sqrt{j+1}}{\sqrt{jN}}\left[e^{-\mathrm{i}\phi_{j+1}}+\frac{1}{j+1}\sum_{k=j+2}^{N}e^{-\mathrm{i}\phi_k}\right],
%\end{equation*}
for $j=1,2,\ldots,N-2$, and $\alpha_{N-1}=-\frac{1}{\sqrt{N-1}}e^{-\mathrm{i}\phi_N}$.
%\begin{equation*}
%\alpha_{N-1}=-\frac{1}{\sqrt{N-1}}e^{-\mathrm{i}\phi_N}.
%\end{equation*}
Therefore, the subradiant state given in \eqref{feb23-1} is a linear combination of the $N-1$ eigenvectors of the DFS given in \eqref{feb22-2}. In this case, the Tavis-Cummings model is neither reachable by its input nor detectable by its output. Thus, the two-level atoms are decoupled from the cavity mode and the external field, and  the Tavis-Cummings model initialized in $|D_N0\rangle$ cannot emit a photon into the cavity or even the external field, and remains in the subradiant state \eqref{feb22-4}. %More discussions on DFSs of quantum linear systems can be found in \cite{NY14tac,NY14prx,ZGPG18} and references therein.
\end{remark}

%%%%%%%%%%%%%%%%%%%%%%%%%%
%%%%%%%%%%%%%%%%%%%%%%%%%%
%%%%%%%%%%%%%%%%%%%%%%%%%%
\section{The multi-excitation case}\label{sec:excitation}

In Section \ref{sec:single_excitation}, we studied the single-excitation dynamics of the Tavis-Cummings model. In this section, we present numerical studies of the multi-excitation scenario. For ease of representation, we calculate the excitation probabilities of the {\it first} two-level atom.

%%%%%%%%%%%%%%%%%%%%%%%%%%
%%%%%%%%%%%%%%%%%%%%%%%%%%
\subsection{The reduced density matrix of the first two-level atom}\label{sec:density}
In  this subsection,  we present the master equation for the Tavis-Cummings model; we also give the expression of the reduced density of the first two-level atom.

%Denote the expectation
%\begin{equation}
%\mathrm{Tr}[\bar{\rho}^{mn}(t)X]=\langle\zeta m_\xi|\mathfrak{j}_t(X)|\zeta n_\xi\rangle, ~ m,n=0,1,
%\end{equation}
%where $X$ is an arbitrary system operator, and $|n_\xi\rangle$ is the vacuum input $\ket{\Phi_0}$ when $n=0$, or the single-photon input state $|1_\xi\rangle$ when $n=1$, and the initial state of the Tavis-Cummings model is $|\zeta\rangle$. It should be noted that in this section $\ket{\zeta}$ is not necessarily of the form  \eqref{feb3:initial} as some atoms may be allowed to be initially in its excited state.

%\begin{equation}
%d\mathfrak{j}_t(X)=\mathfrak{j}_t(\mathcal{L}X)dt+dB_{\mathrm{in}}^\dagger(t)\mathfrak{j}_t([X,L])+\mathfrak{j}_t([L^\dagger,X])dB(t),
%\end{equation}
%where $\mathcal{L}X=-i[X,H]+L^\dagger XL-\frac{1}{2}L^\dagger LX-\frac{1}{2}XL^\dagger L$.

According to  \cite{GJNC12},  the master equation of the Tavis-Cummings model driven by a single-photon input $\ket{1_\xi}$ is
\begin{equation}\label{master}\begin{aligned}
\dot{\bar{\rho}}^{11}(t)&=\mathcal{L}^\star\bar{\rho}^{11}(t)+\xi(t)[\bar{\rho}^{01}(t),L^\ast]+\xi^\ast(t)[L,\bar{\rho}^{10}(t)], \\
\dot{\bar{\rho}}^{10}(t)&=\mathcal{L}^\star\bar{\rho}^{10}(t)+\xi(t)[\bar{\rho}^{00}(t),L^\ast], \\
\dot{\bar{\rho}}^{01}(t)&=\mathcal{L}^\star\bar{\rho}^{01}(t)+\xi^\ast(t)[L,\bar{\rho}^{00}(t)], \\
\dot{\bar{\rho}}^{00}(t)&=\mathcal{L}^\star\bar{\rho}^{00}(t), \ t\geq t_0,
\end{aligned}\end{equation}
where the initial states are
\begin{equation} \label{eq:jun7_initial}
\bar{\rho}^{11}(t_0)=\bar{\rho}^{00}(t_0)=|\zeta\rangle\langle\zeta|, ~ \bar{\rho}^{10}(t_0)=\bar{\rho}^{01}(t_0)=0
\end{equation}
with $\ket{\zeta}$ being the initial system state.  As we will focus on the excitation probability of the first two-level atom, we use the partial trace to get its reduced density operator
\begin{equation}\label{sep5}\begin{aligned}
\rho_{A_1}(t)\triangleq&\ \mathrm{Tr}_{A_N}[\cdots\mathrm{Tr}_{A_2}[\mathrm{Tr}_{\mathrm{cav}}[\bar{\rho}^{11}(t)]]] \\
=&\sum\langle z_2z_3\ldots z_Nn|\bar{\rho}^{11}(t)|z_2z_3\ldots z_Nn\rangle,
\end{aligned}\end{equation}
where $|z_j\rangle=|g_j\rangle$ or $|e_j\rangle$, $j=2,3,\ldots,N$, and $\ket{n}$ is the state of the cavity.

%%%%%%%%%%%%%%%%%%%%%%%%%%
\begin{remark}
When the input is in the vacuum state $\ket{\Phi_0}$,  \eqref{master}-\eqref{eq:jun7_initial} reduce to the commonly used master equation
\begin{equation}\label{eq:master_vac}
\dot{\bar{\rho}}^{00}(t)=\mathcal{L}^\star\bar{\rho}^{00}(t), \ \ \ \bar{\rho}^{00}(t_0) = \ket{\zeta}\bra{\zeta}, \ t\geq t_0.
\end{equation}
 Accordingly, the excitation probability of the first two-level atom is given by $\braket{e_1|\rho_{A_1}(t)|e_1}$, where  $\rho_{A_1}(t)$ should be computed via  \eqref{sep5} by replacing $\bar{\rho}^{11}(t)$ with $\bar{\rho}^{00}(t)$.
\end{remark}

In the following simulations, we assume that there are three resonant two-level atoms that are equally coupled to the cavity, i.e., $\omega_1=\omega_2= \omega_3=\omega_r=1$ and $\Gamma_1=\Gamma_2=\Gamma_3=1$. Also let $\kappa=1.5$.  When the input field is in the single photon state $\ket{\Phi_1}$, we use a Gaussian pulse shape
\begin{equation}\label{sep5a}
\xi(t)=\left(\frac{\Omega^2}{2\pi}\right)^{1/4}\exp\left[-\frac{\Omega^2}{4}(t-t_p)^2\right],
\end{equation}
where $t_p$ is the peak arrival time of the photon, and $\Omega$ is the frequency bandwidth.  Set $t_p=3$ and $\Omega=2\kappa$. See the green solid curve in Fig. \ref{single+excited} for the plot of  $|\xi(t)|^2$. As $|\xi(t)|^2\approx 0$ when $t=0$, it is safe to let the initial time $t_0$ be 0.

\subsection{One initially excited atom plus a single-photon input}\label{subsec:1+1}

Assume that one of the two-level atoms is in the excited state, the cavity is empty, and the system is driven by a single-photon input. Then there are two excitations in the whole system. According to the reduced density matrix \eqref{sep5}, the excitation probability of the first two-level atom in this  two-excitation case can be calculated as
\begin{equation*}\label{sep18}\begin{aligned}
&P_{\mathrm{TLS}1}(t) =\braket{e_1|\rho_{A_1}(t)|e_1} \\
=&\langle e_1g_2g_30|\bar{\rho}^{11}(t)|e_1g_2g_30\rangle+\langle e_1e_2g_30|\bar{\rho}^{11}(t)|e_1e_2g_30\rangle \\
&+\langle e_1g_2e_30|\bar{\rho}^{11}(t)|e_1g_2e_30\rangle+\langle e_1g_2g_31|\bar{\rho}^{11}(t)|e_1g_2g_31\rangle,
\end{aligned}\end{equation*}
where the term $\langle e_1g_2g_30|\bar{\rho}^{11}(t)|e_1g_2g_30\rangle$ indicates that a photon is in the external field.

\begin{figure}[htp!]
\centering
\includegraphics[scale=0.44]{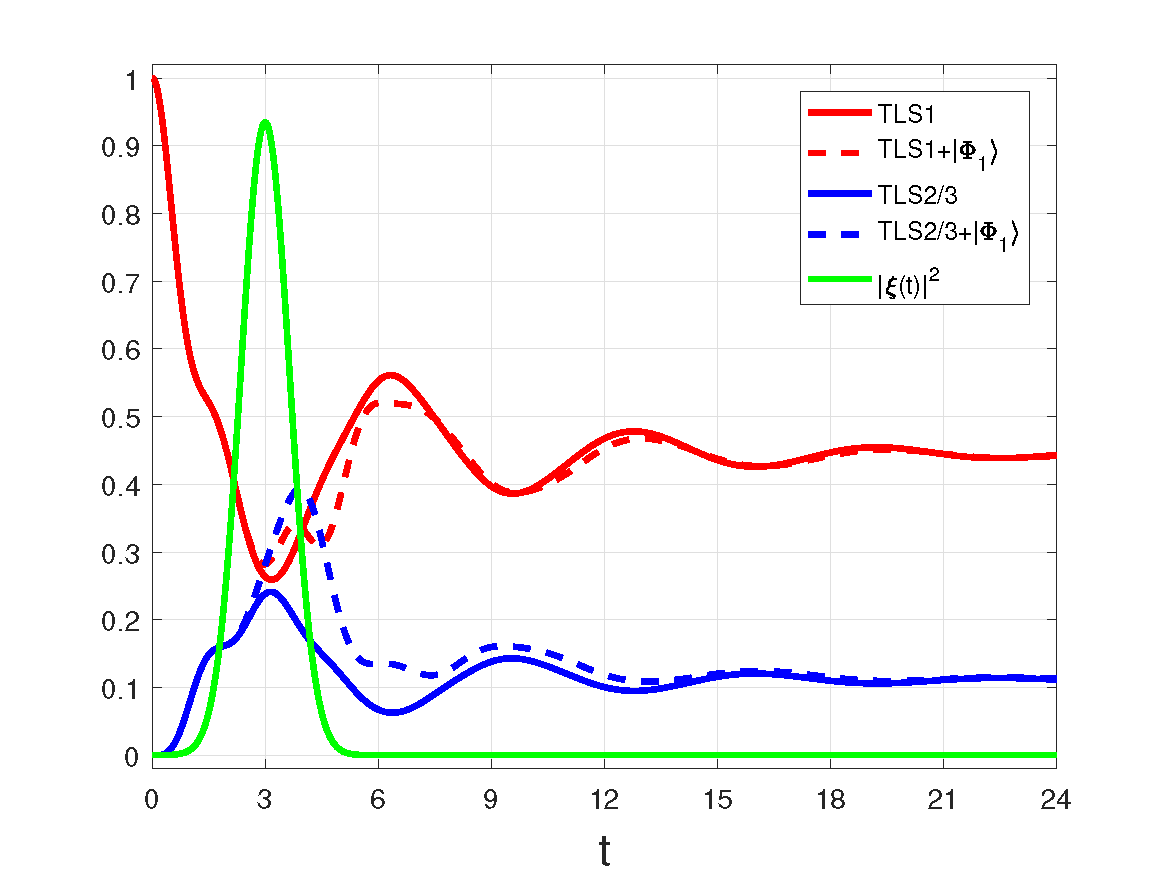}
\caption{The excitation probabilities of the first two-level atom. Green solid curve: $|\xi(t)|^2$;  red solid curve (TLS1: $\ket{e_1g_2g_30}\otimes \ket{\Phi_0}$); blue solid curve  (TLS2/3: $\ket{g_1e_2g_30}\otimes \ket{\Phi_0}$ or $\ket{g_1g_2e_30}\otimes \ket{\Phi_0}$); red dashed curve (TLS1+$\ket{\Phi_1}$: $\ket{e_1g_2g_30}\otimes \ket{\Phi_1}$); blue dashed curve (TLS2/3+$\ket{\Phi_1}$: $\ket{g_1e_2g_30}\otimes \ket{\Phi_1}$ or $\ket{g_1g_2e_30}\otimes \ket{\Phi_1}$).}
\label{single+excited}
\end{figure}

We consider the effect of single-photon input state on the excitation probability, which is simulated in Fig. \ref{single+excited}. The red solid curve is for the case when the initial joint system-field state $\ket{\Psi(0)}=\ket{e_1g_2g_30}\otimes\ket{\Phi_0}$, the blue solid curve is for  $\ket{\Psi(0)}=\ket{g_1e_2g_30}\otimes\ket{\Phi_0}$ or $\ket{\Psi(0)}=\ket{g_1g_2e_30}\otimes\ket{\Phi_0}$, the   red dashed curve is for $\ket{\Psi(0)}=\ket{e_1g_2g_30}\otimes\ket{\Phi_1}$, and  the blue dashed curve is for $\ket{\Psi(0)}=\ket{g_1e_2g_30}\otimes\ket{\Phi_1}$ or $\ket{\Psi(0)}=\ket{g_1g_2e_30}\otimes\ket{\Phi_1}$. Firstly, the incident single photon hardly interacts with the first two-level atom when the time $t$ is far earlier than the peak arrival time $t_p=3$, which is similar to the case when a single photon of Gaussian pulse shape is used to excite a two-level atom; see \cite[Fig. 4(a)]{WMS+11}.  Secondly, the solid red curve is above the solid blue cure for all time. However there is crossover between the dashed curves. Thirdly, the input photon does not affect the steady-state excitation probability of the first TLS (TLS1); however, it does affect the transient dynamics. Finally, from the dashed curves it can be seen that the incident photon tends to increase (or decrease) the excitation probability of the first two-level atom when it is initialized in the ground (or excited) state. %Finally, it can be observed that eventually, the two-level atom does not settle to its ground state (approximately 0.44 for the red curves and 0.11 for the blue curves).

\subsection{Two initially excited atoms plus vacuum input}\label{subsec:2+0}

%\subsubsection{Two initially excited states plus vacuum input}

\begin{figure}[htp!]
\centering
\includegraphics[scale=0.43]{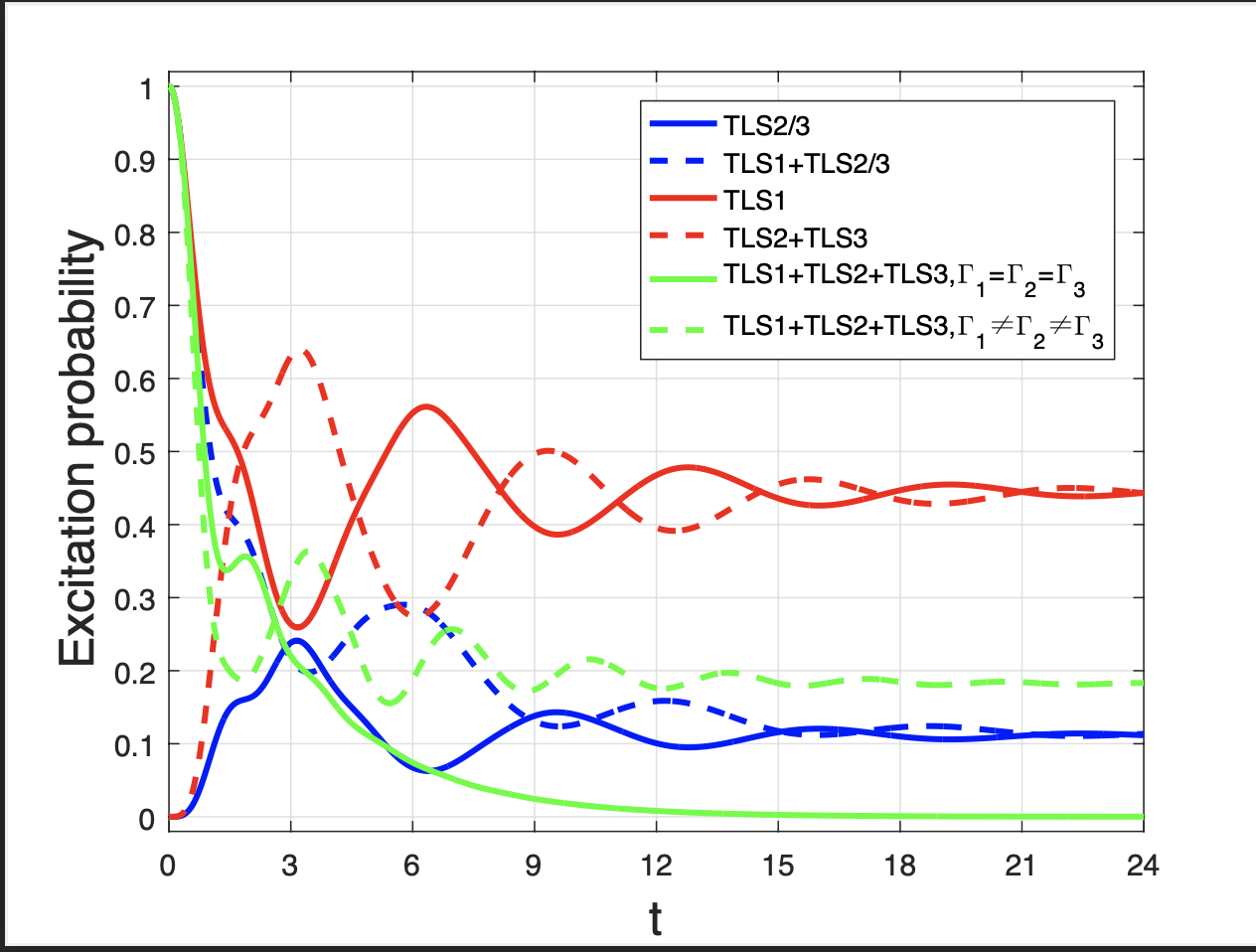}
\caption{The excitation probabilities of the first two-level atom. Red solid curve (TLS1: $\ket{e_1g_2g_30}\otimes\ket{\Phi_0}$);
blue solid curve (TLS2/3: $\ket{g_1e_2g_30}\otimes\ket{\Phi_0}$ or $\ket{g_1g_2e_30}\otimes\ket{\Phi_0}$);  red dashed  curve (TLS2 + TLS3: $\ket{g_1e_2e_30}\otimes\ket{\Phi_0}$);  blue dashed curve (TLS1+ TLS2/3: $\ket{e_1e_2g_30}\otimes\ket{\Phi_0}$ or $\ket{e_1g_2e_30}\otimes\ket{\Phi_0}$);  green solid curve (TLS1+TLS2 + TLS3:$\ket{e_1e_2e_30}\otimes\ket{\Phi_0}$; green dashed curve (TLS1+TLS2+TLS3: $\ket{e_1e_2e_30}\otimes\ket{\Phi_0}$ and $\Gamma_1\neq \Gamma_2\neq \Gamma_3$).
}
\label{two_excited_new}
\end{figure}

In the following simulations, we assume that the Tavis-Cummings  system is driven by  the vacuum input $\ket{\Phi_0}$ and one, two, or even three atoms are initially excited.

We present the following theoretical result first.

%%%%%%%%%%%%%%%%%%%%%%%%%%
\begin{theorem}\label{thm:ss_photon}
Suppose that the Tavis-Cummings system \eqref{system1} is initialized in the state $\ket{\zeta} = \ket{e_1\ldots e_N0}$ and driven by the vacuum input field.  If $\omega_1=\ldots =\omega_N$ and $\Gamma_1=\ldots =\Gamma_N$, then the steady-state ($t = \infty$) output field must be in an $N$-photon state.
\end{theorem}

%%%%%%%%%%%%%%%%%%%%%%%%%%
\begin{proof}
We use the master equation \eqref{eq:master_vac} to prove this result.  The steady-state system state can be found by solving
\begin{equation} \label{eq:master_vac_2}
\mathcal{L}^\star\bar{\rho}^{00}(\infty)=0.
\end{equation}
Clearly, in the steady state ($t=\infty$), the cavity cannot contain any photon; otherwise the photon leaks out from the lossy cavity ($\kappa \neq 0$). In other words, such a state cannot be a steady state. Consequently, the steady-state system state is of the form $\bar{\rho}^{00}(\infty)= \rho_A\otimes \ket{0}\bra{0}$, where $\rho_A$ is the steady state of the two-level atoms. By the form of the system Hamiltonian $H_\mathrm{TC}$ and the coupling operator $L$ given in subsection \ref{sec:des}, \eqref{eq:master_vac_2} is actually
\begin{equation*} \label{eq:master_vac_3}\begin{aligned}
&\frac{1}{2}\sum_j \omega_j [\sigma_{z,j}, \rho_A]  \otimes\ket{0}\bra{0} \\
+&
\sum_j \Gamma_j \bigg(\sigma_{-,j}\rho_A\otimes \ket{1}\bra{0} - \rho_A \sigma_{+,j} \ket{0}\bra{1}  \bigg)=0.
\end{aligned}
\end{equation*}
Given $\omega_1=\ldots =\omega_N\equiv \omega_s$ and $\Gamma_1=\ldots =\Gamma_N\equiv \Gamma$, we have
\begin{equation*}\begin{aligned}
&\frac{\omega_s}{2}\sum_j  [\sigma_{z,j}, \rho_A] \otimes \ket{0}\bra{0}
\\
&+\Gamma\sum_j \bigg(\sigma_{-,j}\rho_A\otimes \ket{1}\bra{0} - \rho_A \sigma_{+,j} \ket{0}\bra{1}  \bigg)=0,
\end{aligned}\end{equation*}
which yields
\begin{equation}  \label{eq:master_vac_4}
\sum_j\sigma_{-,j}\rho_A=0.
\end{equation}
In general, $\rho_A$ is of the form
\[
\sum_{i_1,\ldots,i_N; j_1,\ldots, j_N} \alpha_{i_1,\ldots,i_N; j_1,\ldots, j_N} \ket{f_{i_1}\ldots f_{i_N}}\bra{g_{j_1}\ldots g_{j_N}},
\]
where $f_{i_k}$ and $g_{j_k}$ are either $e_k$ or $g_k$. Because $\omega_1=\ldots =\omega_N$ and $\Gamma_1=\ldots =\Gamma_N$, all the atoms are indistinguishable. As a result, all the coefficients must be identical; i.e., $\alpha_{i_1,\ldots,i_N; j_1,\ldots, j_N}\equiv \alpha$ for some $\alpha$. Consequently,
\begin{equation*}\label{eq:master_vac_5}
\rho_A = \alpha \sum_{i_1,\ldots,i_N; j_1,\ldots, j_N} \ket{f_{i_1}\ldots f_{i_N}}\bra{g_{j_1}\ldots g_{j_N}}.
\end{equation*}
Substituting this form of $\rho_A$ into \eqref{eq:master_vac_4} we get $\sigma_{-,j}\ket{f_{i_1}\ldots f_{i_N}}=0,$ for all $j=1,\ldots,m$.
%\[
%\sigma_{-,j}\ket{f_{i_1}\ldots f_{i_N}}=0, \ \ \forall j=1,\ldots,m.
%\]
Clearly, all $f_{i_k}$  must be $g_k$. Consequently, $\rho_A$ contains a single term and is of the form $\rho_A= \ket{g_1\ldots g_N}\bra{g_1\ldots g_N}$. In other words, all atoms are in their ground state. Therefore, in the steady state, the output field is in an $N$-photon state.
\end{proof}

The simulation results are shown in Fig. \ref{two_excited_new}. We have the following observations. (i) The first atom's excitation probability when it is initialized in the excited state (the red solid curve) is greater than that when either the second atom or the third atom is initialized in the excited state (the blue solid curve). (ii) The first atom's excitation probability when both the second and third atoms are initialized in the excited state (the red dashed curve) is greater than that when both the first and second atoms (or both the first and third atoms) are initialized in the excited state (the blue dashed curve). (iii) The red solid and dashed curves have the same final value $\approx0.44$, while the blue solid and  dashed curves  have the same final value $\approx0.11$.  (iv) The excitation probability of the first two-level atom settles to $0$ when all the three atoms are initially excited, see the green solid curve in Fig. \ref{two_excited_new}. Actually, by Theorem \ref{thm:ss_photon}, all the three atoms eventually settle in their ground state and a 3-photon output state is generated. However, if the coupling strengths are not identical, the atoms may not settle to their ground state, as demonstrated by the green dashed curve in Fig. \ref{two_excited_new}, where $\Gamma_1=1$, $\Gamma_2=1.5$, and $\Gamma_3=2$.

%%%%%%%%%%%%%%%%%%%%%%%%%%
\begin{remark}
In Figs. \ref{single+excited} and \ref{two_excited_new}, when the first atom is initialized in the excited state and the other two are in the ground state, the final value of the excitation probability of the first atom is approximately $0.44$.  By Theorem \ref{state}, $|c_1(\infty)|^2=\frac{4}{9}$, which explains this simulation result. The other final value $0.11$ in Figs. \ref{single+excited} and \ref{two_excited_new} can be explained by  $|c_2(\infty)|^2 = |c_3(\infty)=\frac{1}{9}$. Moreover, $|c_k(t)|^2 \to 1$ and $|c_j(t)|^2 \to 0$ ($j\neq k$) as $N\to \infty$. This means that the system tends to remain intact when the number of atoms is sufficiently large. Finally, if $\Gamma_1 = \cdots = \Gamma_N$, then by \eqref{oct27-4} we have  $|c_{N+1,k}(\infty)|^2=\ff{1}{N}$, i.e., only $1/N$ of the excitation energy is radiated. This confirms the analysis in the second paragraph above (29) in \cite{Dicke54}.
\end{remark}

%%%%%%%%%%%%%%%%%%%%%%%%%%
\begin{remark}
Although the red solid and dashed curves have  identical stationary excitation probability, their crests and troughs are almost symmetrical during the transient process. Similar phenomena can be observed when $N=2$. This is consistent with the experimental result given in \cite[Fig. 4c]{MCG07}. Specifically, write states $|e_1g_2\rangle$ and $|g_1e_2\rangle$ as $\ket{\uparrow\downarrow}$ and
$\ket{\downarrow\uparrow}$, respectively. In  \cite[Fig. 4c]{MCG07}, the solid red (green) circles plot the evolution of $\ket{\downarrow\uparrow}$ ( $\ket{\uparrow\downarrow}$). The coherent state transfer of the two states emerges and oscillates symmetrically.
\end{remark}

\section{A general form of the joint system-field states}\label{sec:alter}

In Section \ref{sec:single_excitation}, analytical results are presented for the single-excitation  Tavis-Cummings model. These results are not applicable to the multi-excitation case discussed in Section \ref{sec:excitation}. Motivated by this, in this section we derive a recursive relation for the computation of the joint system-field state when the system initially contains multiple excitations $R>1$ and is driven by $m$ input fields initialized in the vacuum state. %The description of a general model is given in Subsection  \ref{subsec:model}, the recursive relation for computing the joint state is presented in Subsection \ref{subsec:general form}, and its application to the Tavis-Cummings model is demonstrated in Subsections \ref{subsec:sim1}---\ref{subsec:sim3}.

%%%%%%%%%%%%%%%%%%%%%%%%%%
%%%%%%%%%%%%%%%%%%%%%%%%%%
\subsection{Modeling} \label{subsec:model}

Let $H$ be the internal system Hamiltonian of a quantum system. As the system Hamiltonian may be tuned by a time-varying classical signal, for example, transition frequencies of superconducting qubits can be tuned via Josephson energy in Josephson-junction-based superconducting circuits \cite{BGG+21}, we use $H(t)$ to emphasize the explicit dependence of the system Hamiltonian on time. The coupling between the quantum system and the $k$th ($k=1,\ldots,m$) input channel can be described by a coupling operator $L_k(t)$; again, here we allow explicit time dependence of the coupling as in some quantum systems couplings can be tuned in real time \cite{PSR+14}.

The system of interest may be an ensemble of TLSs, or resonators, or even  an ensemble of TLSs residing in a resonator like the Tavis-Cummings model studied in this paper. If the number of TLSs is finite, the number of photons in the resonator is finite, and the system is driven by a finite number of photons, then the number of total excitations $R$ in the whole system (the system plus the external field) is finite too. In this case, the state of the quantum system  has an orthonormal basis of the form
\begin{equation}\label{basis}
\{\ket{0_s}, \ldots, \ket{{(K-1)}_s}\},
\end{equation}
where the subscript ``$s$" indicates that the  basis states $\ket{j_s}$, $j=0,1,\ldots,K-1$, are for the system. For our Tavis-Cummings mode in Fig. \ref{system}, $K=2^N(R+1)$ as each of the $N$ atoms has two basis states and the cavity can contain at most $R$ photons.

%%%%%%%%%%%%%%%%%%%%%%%%%%
\begin{remark}
If the system of interest is coherently driven, then the number of excitations can be arbitrary large as the drive may generate an arbitrary number of excitations.  However, if the coherent drive is not very strong and lasts not long, the number of excitations in the system will be upper bounded. Hence, it can still be assumed that  the system admits a basis  of finitely many elements.
\end{remark}

For notational convenience, the following notation will be adopted. For each $i=1,\ldots,m$,  we use $t^i_{1\to k_i}$ to denote a set of ordered real numbers $\{t^i_1, \ldots,t^i_{k_i}: t^i_1<\cdots<t^i_{k_i}\}$. Similarly, $1_{t^j_{1\to k_j}}$ is a shorthand of $1_{t^j_1},\cdots, 1_{t^j_{k_j}}$. Let $\int_r^{t^i_{2\to k_i}}$ be the abbreviation for $\int_r^{t^i_{k_i}}\cdots\int_r^{t^i_{2}}$. Finally, for each given non-negative integer $n$, $\sum_{k_1,\ldots,k_m}^n$ means the summation over all  possible combinations of the non-negative integers $k_1,\ldots,k_m$ that satisfy $\sum_{j=1}^m k_j=n$.

The temporal photon-number basis of $n$ photons superposed over $m$ channels is
\begin{equation}\label{n photon basis}\begin{aligned}
\Bigg\{&|1_{t^1_{1}}\rangle\otimes\cdots\otimes|1_{t^1_{k_1}}\rangle\otimes\cdots\otimes|1_{t^m_{1}}\rangle\otimes\cdots\otimes |1_{t^m_{k_m}}\rangle: \\ &t^1_{1}<\cdots<t^1_{k_1},\cdots,t^m_{1}<\cdots<t^m_{k_m},\sum_{i=1}^m k_i =n\Bigg\},
\end{aligned}\end{equation}
where $k_1,\ldots,k_m$ are non-negative integers. Particularly, if $k_i=0$, then there are no photons in channel $i$ and correspondingly the term $|1_{t^i_{1}}\rangle\otimes\cdots\otimes|1_{t^i_{k_i}}\rangle$ reduces to $\ket{\Phi_0}$. As a common practice in quantum physics, the tensor product state in \eqref{n photon basis} is often written as $|1_{t^1_{1}},\ldots, 1_{t^1_{k_1}},\ldots, 1_{t^m_{1}},\ldots, 1_{t^m_{k_m}}\rangle$, which is $\ket{1_{t^1_{1\to k_1}},\ldots,1_{t^m_{1\to k_m}}}$ by means of the notation given in the last paragraph.

Using the temporal photon-number basis in \eqref{n photon basis}, an $m$-channel $n$-photon state can be written as
\begin{equation}\label{eq:n photons}\begin{aligned}
&|n_\xi\rangle=\sum_{k_1,\ldots,k_m}^n\int_{-\infty}^\infty\cdots\int_{-\infty}^\infty \xi(t^1_{1\to k_1},\ldots,t^m_{1\to k_m}) \\
& \ \ \  \ \ \ \  \ \ \  \ \ \ \ \  |1_{t^1_{1\to k_1}},\ldots,1_{t^m_{1\to k_m}}\rangle dt^1_{1\to k_1}\cdots dt^m_{1\to k_m},
\end{aligned}\end{equation}
where the $j$th channel has $k_j$ photons, and $\sum_{j=1}^{m}k_j=n$ is the total photon number.  Clearly, $\xi(t^1_{1\to k_1},\ldots,t^m_{1\to k_m})$ is the probability amplitude of the component that the $j$th channel has photons at time instants $t^j_1,\ldots,t^j_{k_j}$ for all $j=1,\ldots, m$.

Due to the indistinguishability of photons in each channel, for each fixed $j=1,\ldots, m$, $\xi(t^1_{1\to k_1},\ldots,t^m_{1\to k_m})$ is permutation-invariant with respect to indices $\{t^j_{1\to k_j}\}$. Thus, under scaling $\frac{1}{k_1!\cdots k_m!}$, the state $|n_\xi\rangle$ can be rewritten as
\begin{align}
&|n_\xi\rangle \nonumber \\
=&\sum_{k_1,\ldots,k_m}^n\int_{-\infty}^\infty\int_{-\infty}^{t^1_{2\to k_1}} \cdots \int_{-\infty}^\infty \int_{-\infty}^{t^m_{2\to k_m}}
\xi(t^1_{1\to k_1},\ldots,t^m_{1\to k_m})  \nonumber\\
& \ \ \  \ \ \ \  \ \ \   |1_{t^1_{1\to k_1}},\ldots,1_{t^m_{1\to k_m}}\rangle dt^1_{1\to k_1}\cdots dt^m_{1\to k_m}.\label{eq:aug16_5b}
\end{align}

In fact, one can always define an $n$-photon state using \eqref{eq:aug16_5b} with an arbitrary multivariate function $\xi(t^1_{1\to k_1},\ldots,t^m_{1\to k_m})$ (under normalization), then obtain the form in \eqref{eq:n photons} by permutating indices $\{t^j_{1\to k_j}\}$ for all $j=1,\ldots,m$. The form of $n$-photon states \eqref{eq:aug16_5b} is more often to see in practice than that in (\ref{eq:n photons}). For example, by applying a coherent drive to a superconducting qubit embedded in a chiral waveguide, the qubit may generate photon states of the form \eqref{eq:aug16_5b}.  Therefore, in what follows we use \eqref{eq:aug16_5b} to describe $n$-photon states.

%%%%%%%%%%%%%%%%%%%%%%%%%%
%%%%%%%%%%%%%%%%%%%%%%%%%%
\subsection{The general form}\label{subsec:general form}
In this subsection, we present a computational procedure that can be used to compute the joint system-field state. Some other computational framework can be found in e.g., \cite{FTR+18,TFX+18}.

The basis state of the joint system-field state when the system is at level $\ket{j_s}$ and  channel $i$ has $k_i$ photons at time instants $t^i_{1\to k_i}$ is $\ket{j_s} \otimes |1_{t^1_{1\to k_1}},\ldots,1_{t^m_{1\to k_m}}\rangle$.
%\begin{equation*}
%\ket{j_s} \otimes |1_{t^1_{1\to k_1}},\ldots,1_{t^m_{1\to k_m}}\rangle.
%\end{equation*}
 Hence, the joint system-field state is of the general form
\begin{equation} \label{eq:Psi_june9}
\begin{aligned}
\ket{\Psi}
=&\sum_{j=0}^{K-1}\sum_{n=0}^\infty \sum_{k_1,\ldots,k_m}^n \int_{-\infty}^\infty \int_{-\infty}^{t^1_{2\to k_1}} \cdots \int_{-\infty}^\infty \int_{-\infty}^{t^m_{2\to k_m}}
\\
&\hspace{4ex}\xi^{j_s}(t^1_{1\to k_1},\ldots,t^m_{1\to k_m})\ket{j_s} \\
&\hspace{4ex} dB_{\mathrm{in},1}^\ast(t^1_{1\to k_1})\cdots dB_{\mathrm{in},m}^\ast(t^m_{1\to k_m})\ket{\Phi_0},
\end{aligned}
\end{equation}
where we write informally $d B_{\mathrm{in}, k}^\ast(t)=b_{\mathrm{in}, k}^\ast(t)dt$ and $dB_{\mathrm{in},i}^\ast(t^i_{1\to k_i})$ as the product $dB^\ast_{\mathrm{in},i}(t^i_1) \cdots dB^\ast_{\mathrm{in},i}(t^i_{k_i}) $.  The normalization condition $\braket{\Psi|\Psi}=1$ gives
\begin{equation*}\begin{aligned}
&\sum_{j=0}^{K-1}\sum_{n=0}^\infty \sum_{k_1,\ldots,k_m}^n
\int_{-\infty}^\infty \int_{-\infty}^{t^1_{2\to k_1}} \cdots \int_{-\infty}^\infty \int_{-\infty}^{t^m_{2\to k_m}}
\\
&
\hspace{4ex} |\xi^{j_s}(t^1_{1\to k_1},\ldots,t^m_{1\to k_m})|^2  dt^1_{1\to k_1}\cdots dt^m_{1\to k_m}
=1.
\end{aligned}
\end{equation*}
Denote
\begin{equation*}
\ket{\eta(t^1_{1\to k_1}, \ldots, t^m_{1\to k_m}) }
=\sum_{j=0}^{K-1} \xi^{j_s}(t^1_{1\to k_1}, \ldots, t^m_{1\to k_m}) \ket{j_s}.
\end{equation*}
It is worthwhile to note that $\ket{\eta(t^1_{1\to k_1}, \ldots, t^m_{1\to k_m}) } $ in general is not normalized.
The state $\ket{\Psi}$ in  \eqref{eq:Psi_june9} can be re-written as
\begin{equation*}
\begin{aligned}
\ket{\Psi}
=&\sum_{n=0}^\infty \sum_{k_1,\ldots,k_m}^n \int_{-\infty}^\infty \int_{-\infty}^{t^1_{2\to k_1}} \cdots \int_{-\infty}^\infty \int_{-\infty}^{t^m_{2\to k_m}} \\
&\hspace{4ex}\ket{\eta(t^1_{1\to k_1},\ldots,t^m_{1\to k_m})} \\
&\hspace{4ex}dB_{\mathrm{in},1}^\ast(t^1_{1\to k_1})\cdots dB_{\mathrm{in},m}^\ast(t^m_{1\to k_m})\ket{\Phi_0}.
\end{aligned}
\end{equation*}

If the input field is initially in the vacuum state $\ket{\Phi_0}$, then the initial joint system-field state is $\ket{\Psi(t_0)}=\sum_{j=0}^{K-1} \xi^{j_s} \ket{j_s}\otimes \ket{\Phi_0}  \equiv  \ket{\eta_{0}} \otimes \ket{\Phi_0}$,
%\begin{equation*} \label{eq:jun9_Psi_t0}
%\ket{\Psi(t_0)}=\sum_{j=0}^{K-1} \xi^{j_s} \ket{j_s}\otimes \ket{\Phi_0}  \equiv  \ket{\eta_{0}} \otimes \ket{\Phi_0},
%\end{equation*}
where $\sum_{j=0}^{K-1} |\xi^{j_s}|^2 =1$. Because the  photon generation time is from the initial time $t_0=0$ to the present time $t>0$,  the joint system-field state at time $t$ is
\begin{equation}\label{Psi_3}\begin{aligned}
\ket{\Psi(t)}
=&\sum_{n=0}^\infty \sum_{k_1, \ldots, k_m}^n \int_{0}^t  \int_{0}^{t^1_{2\to k_1}} \cdots \int_{0}^t  \int_{0}^{t^m_{2\to k_m}} \\ &\ket{\eta_t(t^1_{1\to k_1}, \ldots, t^m_{1\to k_m}) } \\
&dB_{\mathrm{in},1}^\ast(t^1_{1\to k_1})\cdots dB_{\mathrm{in},m}^\ast(t^m_{1\to k_m})\ket{\Phi_0},
\end{aligned}\end{equation}
where
\begin{equation} \label{eq:jun9_temp1}
\ket{\eta_t(t^1_{1\to k_1}, \ldots, t^m_{1\to k_m}) }
=\sum_{j=0}^{K-1} \xi^{j_s}_t(t^1_{1\to k_1}, \ldots, t^m_{1\to k_m}) \ket{j_s}.
\end{equation}
(Here, the subscript ``$t$'' indicates that the coefficients are time dependent.)
In particular, the  term in $\ket{\Psi(t)}$ that corresponds to $n=0$ is  $\ket{\eta_t}\ket{\Phi_0} =  \sum_{j=0}^{K-1} \xi^{j_s}_t  \ket{j_s}\ket{\Phi_0}$.
%\begin{equation*}\label{eq:n=0}
%\ket{\eta_t}\ket{\Phi_0} =  \sum_{j=0}^{K-1} \xi^{j_s}_t  \ket{j_s}\ket{\Phi_0}.
%\end{equation*}
In other words, the field is in the vacuum state and the system is in the superposition state $  \sum_{j=0}^{K-1} \xi^{j_s}_t  \ket{j_s}$.

%\begin{remark}
%The subscript ``$t$'' indicates the present time.
%\end{remark}
In what follows, we aim to derive formulas for computing the joint system-field state $\ket{\Psi(t)}$.  Differentiating both sides of \eqref{Psi_3} yields
%\vspace{-2.5ex}
\begin{equation}\label{sep28_d_Psi}\begin{aligned}
&d\ket{\Psi(t)} \\
=&\sum_{n=1}^\infty \sum_{k_1, \ldots, k_m}^n dB_{\mathrm{in},1}^\ast (t) \int_{0}^t \int_{0}^{t^1_{2\to k_1-1}} \cdots \int_{0}^t \int_{0}^{t^m_{2\to k_m}} \\
&\hspace{12ex}\ket{\eta_t(t^1_{1\to k_1-1},t, t^m_{2\to k_2},\ldots, t^m_{1\to k_m}) } \\
&\hspace{12ex}dB_{\mathrm{in},1}^\ast(t^1_{1\to k_1-1})\cdots dB_{\mathrm{in},m}^\ast(t^m_{1\to k_m})\ket{\Phi_0} \\
&\vdots \\
&+\sum_{n=1}^\infty \sum_{k_1, \ldots, k_m}^n dB_{\mathrm{in},m}^\ast (t) \int_{0}^t \int_{0}^{t^1_{2\to k_1}} \cdots \int_{0}^t \int_{0}^{t^m_{2\to k_m-1}} \\
&\hspace{14ex}\ket{\eta_t(t^1_{1\to k_1}, \ldots, t^{m-1}_{1\to k_{m-1}}, t^m_{1\to k_m-1},t)} \\
&\hspace{14ex}dB_{\mathrm{in},1}^\ast(t^1_{1\to k_1})\cdots dB_{\mathrm{in},m}^\ast(t^m_{1\to k_m-1})\ket{\Phi_0} \\
&+dt\sum_{n=0}^\infty \sum_{k_1, \ldots, k_m}^n \int_{-\infty}^t \int_{-\infty}^{t^1_{2\to k_1}} \cdots \int_{-\infty}^t \int_{-\infty}^{t^m_{2\to k_m}} \\
&\hspace{14ex}\ket{\dot{\eta}_t(t^1_{1\to k_1}, \ldots, t^m_{1\to k_m})} \\
&\hspace{14ex}dB_{\mathrm{in},1}^\ast(t^1_{1\to k_1})\cdots dB_{\mathrm{in},m}^\ast(t^m_{1\to k_m})\ket{\Phi_0},
\end{aligned}\end{equation}
where $\dot{\eta}_t$ means the derivative of $\eta$ with respect to $t$. Comparing the right-hand sides of \eqref{sep28_d_Psi} and \eqref{SME_b} (for $\ket{\Psi(t)}$ in  \eqref{Psi_3}), we get
\vspace{-2.5ex}
\begin{equation}\label{dot_eta}\begin{aligned}
&\ket{\dot{\eta}_t(t^1_{1\to k_1},\ldots, t^m_{1\to k_m})} \\
=&-\mathrm{i}H_\mathrm{eff}(t)\ket{\eta_t(t^1_{1\to k_1},\ldots, t^m_{1\to k_m})},
\end{aligned}\end{equation}
and
\vspace{-2.5ex}
\begin{equation}\label{+1}\begin{aligned}
&\ket{\eta_t(t^1_{1\to k_1},\ldots,t^i_{1\to
k_i},t,\ldots,t^m_{1\to k_m})} \\
=&L_i(t)\ket{\eta_t(t^1_{1\to k_1},\ldots,t^i_{1\to k_i},\ldots,t^m_{1\to k_m})}
\end{aligned}\end{equation}
for all $i=1,\ldots,m$. Under the basis \eqref{basis}, the effective Hamiltonian $H_\mathrm{eff}$ of the system has a matrix representation. Hence, similar to what has been done in \cite{LZW20}, define the propagator $V(t)$ (a matrix function) that solves a system of deterministic homogeneous ODEs $\dot{V}(t)=-\mathrm{i}H_\mathrm{eff}(t)V(t)$
%\begin{equation*}\label{jan27-1}
%\dot{V}(t)=-\mathrm{i}H_\mathrm{eff}(t)V(t)
%\end{equation*}
under the initial condition $V(0)=I$. Then one can define the transition matrix
%\vspace{-4ex}
\begin{equation}\label{jan27-2}
G(t,\tau)\triangleq V(t)V(\tau)^{-1}, ~ t,\tau\geq 0.
\end{equation}
Clearly, for the $0$ photon state case, \eqref{dot_eta} and \eqref{jan27-2} yield
%\vspace{-4ex}
\begin{equation} \label{eq:jun9_initial}
\ket{\eta_t}= G(t,0)\ket{\eta_{0}}, ~ t\geq 0.
\end{equation}
As is well-known in linear systems theory, see e.g. \cite[Chapter 4]{Rugh96},  iteratively using the transition matrix $G(t,\tau)$ in  \eqref{jan27-2} and  \eqref{+1} we have
\begin{equation*}\begin{aligned}
&\ket{\eta_t(t^1_{1\to k_1}, t^2_{1\to k_2}, \ldots, t^m_{1\to k_m})} \\
=&\ G(t,t^1_{k_1})\ket{\eta_{t^1_{k_1}}(t^1_{1\to k_1-1}, t^1_{k_1}, t^2_{1\to k_2}, \ldots, t^m_{1\to k_m})} \\
=&\ G(t,t^1_{k_1})L_1(t^1_{k_1})\ket{\eta_{t^1_{k_1}}(t^1_{1\to k_1-1}, t^2_{1\to k_2}, \ldots, t^m_{1\to k_m})} \\
=&\ G(t,t^1_{k_1})L_1(t^1_{k_1})G(t^1_{k_1},t^1_{k_1-1}) \\
&\times\ket{\eta_{t^1_{k_1-1}}(t^1_{1\to k_1-1}, t^2_{1\to k_2}, \ldots, t^m_{1\to k_m})} \\
&\vdots \\
=&\ G(t,t^1_{k_1})L_1(t^1_{k_1})\cdots G(t^1_2,t^1_1)L_1(t^1_1) \\
&\times G(t^1_1,t^2_{k_2})\ket{\eta_{t^2_{k_2}}(t^2_{1\to k_2}, \ldots, t^m_{1\to k_m})} \\
&\vdots \\
=&\ G(t,t^1_{k_1})L_1(t^1_{k_1})\cdots G(t^1_2,t^1_1)L_1(t^1_1)G(t^1_1,t^2_{k_2})L_2(t^2_{k_2}) \times \cdots \\
&\times G(t^2_2,t^2_1)L_2(t^2_1) G(t^2_1,t^3_{k_3}) \ket{\eta_{t^3_{k_3}}(t^3_{1\to k_3}, \ldots, t^m_{1\to k_m})} \\
&\vdots \\
=&\ G(t,t^1_{k_1})L_1(t^1_{k_1})\cdots G(t^1_2,t^1_1)L_1(t^1_1) \\
&\times G(t^1_1,t^2_{k_2})L_2(t^2_{k_2})\cdots G(t^2_2,t^2_1)L_2(t^2_1) \times \cdots \\
&\times G(t^{m-1}_1,t^m_{k_m})L_m(t^m_{k_m})\cdots G(t^m_{2},t^m_1)L_m(t^m_1) G(t^m_1,0)\ket{\eta_{0}}.
\end{aligned}\end{equation*}
Consequently, we have the following general formula
\begin{equation}\label{general}\begin{aligned}
&\ket{\eta_t(t^1_{1\to k_1}, t^2_{1\to k_2}, \ldots, t^m_{1\to k_m})} \\
=&G(t,t^1_{k_1})L_1(t^1_{k_1})\cdots G(t^1_2,t^1_1)L_1(t^1_1) \\
&\times G(t^1_1,t^2_{k_2})L_2(t^2_{k_2})\cdots G(t^2_2,t^2_1)L_2(t^2_1)\cdots \\
&\times G(t^{m-1}_1,t^m_{k_m})L_m(t^m_{k_m})\cdots G(t^m_{2},t^m_1)L_m(t^m_1) \\
&\times G(t^m_1,0)\ket{\eta_{0}}.
\end{aligned}\end{equation}

\begin{remark}
To better understand the recursive algorithm given in \eqref{general}, we take the two-channel case as an example, i.e., $m=2$. In this case, the first channel contains $k_1$ photons and the second channel contains $k_2$ photons. The recursive functions \eqref{dot_eta} and \eqref{+1} are reduced to be
\begin{equation}\label{oct22-8}\begin{aligned}
\ket{\dot{\eta}_t^{k_1,k_2}(t^1_{1\to k_1},t^2_{1\to k_2})}=-\mathrm{i}H_\mathrm{eff}(t)\ket{\eta_t^{k_1,k_2}(t^1_{1\to k_1}, t^2_{1\to k_2})},
\end{aligned}\end{equation}
and
\begin{equation}\label{oct22-9}\begin{aligned}
&\ket{\eta_t^{k_1+1,k_2}(t^1_{1\to k_1},t,t^2_{1\to k_2})}=L_1(t)\ket{\eta_t^{k_1,k_2}(t^1_{1\to k_1},t^2_{1\to k_2})}, \\
&\ket{\eta_t^{k_1,k_2+1}(t^1_{1\to k_1},t^2_{1\to k_2},t)}=L_2(t)\ket{\eta_t^{k_1,k_2}(t^1_{1\to k_1},t^2_{1\to k_2})}.
\end{aligned}\end{equation}
According to \eqref{general}, we have
\begin{equation}\label{oct22-10}\begin{aligned}
&\ket{\eta_t^{k_1,k_2}(t^1_{1\to k_1},t^2_{1\to k_2})} \\
=&G(t,t^1_{k_1})L_1(t^1_{k_1})\cdots G(t^1_2,t^1_1)L_1(t^1_1) \\
&\times G(t^1_1,t^2_{k_2})L_2(t^2_{k_2})\cdots G(t^2_2,t^2_1)L_2(t^2_1) \\
&\times G(t^2_1,0)\ket{\eta_{0}},
\end{aligned}\end{equation}
which yields
\begin{equation}\label{oct22-11}\begin{aligned}
\ket{\eta_t}=&G(t,0)\ket{\eta_{0}}, \\
\ket{\eta_t^{1,0}(t^1_1)}=&G(t,t_1^1)L_1(t^1_1)G(t_1^1,0)\ket{\eta_{0}}, \\
\ket{\eta_t^{0,1}(t^1_2)}=&G(t,t^2_1)L_2(t^2_1)G(t^2_1,0)\ket{\eta_{0}}, \\
\ket{\eta_t^{1,1}(t^1_1,t^2_1)}=&G(t,t^1_1)L_1(t^1_1)G(t^1_1,t^2_1)L_2(t^2_1)G(t^2_1,0)\ket{\eta_{0}} \\
=&G(t,t^2_1)L_2(t^2_1)G(t^2_1,t^1_1)L_1(t^1_1)G(t^1_1,0)\ket{\eta_{0}}, \\
&\vdots
\end{aligned}\end{equation}
That is, there are multiple sequential orders to generate a ($k_1+k_2$)-photon state distributed in two channels ($k_1,k_2\geq1$), which allow the two channels being created crosswise.
\end{remark}

The following recursive relation turns out useful for deriving all the states.

\begin{equation}\label{oct23_+1}\begin{aligned}
&\ket{\eta_t(t^1_{1\to k_1},\ldots,t^i_{1\to k_i+1},\ldots,t^m_{1\to k_m})} \\
=&\  G(t,t^i_{k_i+1})L_i(t^i_{k_i+1})G(t,t^i_{k_i+1})^{-1} \\
&\ket{\eta_t(t^1_{1\to k_1},\ldots,t^i_{1\to k_i},\ldots,t^m_{1\to k_m})}.
\end{aligned}\end{equation}

%%%%%%%%%%%%%%%%%%%%%%%%%%%
%\begin{algorithm}
%\caption{Emitting photons in $m$ channels from the system state $|\eta_{t}\rangle$.}
%\label{algorithm1}
%\textbf{Step 1} Let the current system state be
%\begin{equation}
%|\eta_{t}\rangle=\ket{\eta_t(t^1_{1\to k_1},\ldots,t^i_{1\to k_i},\ldots,t^m_{1\to k_m})},
%\end{equation}
%i.e., the $m$ channels contain $k_1,\ldots,k_m$ photons, respectively.
%
%\textbf{Step 2} Calculating the transition matrix $G(t,\tau)$ by \eqref{jan27-1} and \eqref{jan27-2}.
%
%\textbf{Step 3} (Iterating equation) \textbf{for} $i=1,2,\ldots,m$, \textbf{do}
%\begin{equation}\label{oct23_+1}\begin{aligned}
%&\ket{\eta_t(t^1_{1\to k_1},\ldots,t^i_{1\to k_i+1},\ldots,t^m_{1\to k_m})} \\
%=&G(t,t^i_{k_i+1})L_i(t^i_{k_i+1})G(t,t^i_{k_i+1})^{-1} \\
%&\ket{\eta_t(t^1_{1\to k_1},\ldots,t^i_{1\to k_i},\ldots,t^m_{1\to k_m})}.
%\end{aligned}\end{equation}
%\textbf{end for}.
%
%\textbf{return}: The emitted multi-channel multi-photon state $\ket{\eta_t^{k_1^\prime,\ldots,k_m^\prime}(t^1_{1\to k_1^\prime},\ldots,t^m_{1\to k_m^\prime})}$, i.e., the $m$ channels contain $k_1^\prime,\ldots,k_m^\prime$ photons, respectively. $k_i^\prime\geq k_i$, $i=1,\ldots,m$.
%\end{algorithm}

%%%%%%%%%%%%%%%%%%%%%%%%%%
\begin{remark}\label{rem:for thm 5.1}
Theorem \ref{state} in Section \ref{sec:super} can be proved by applying the recursive relation \eqref{oct23_+1} to the special case when $m=1$ and the initial system-field state is
  \begin{equation}\label{oct21-2}
\ket{\Psi(0)}=\ket{\Psi_k(0)} =|\eta_{0}\rangle\otimes\ket{\Phi_0},
\end{equation}
 where $|\eta_{0}\rangle=|g_1\cdots e_k\cdots g_N0\rangle$.
%\begin{equation*}\label{oct21-3}\begin{aligned}
%|\eta_{0}\rangle=|g_1\cdots e_k\cdots g_N0\rangle.
%\end{aligned}\end{equation*}
In what follows we sketch the proof. Because there is only one input channel and the number of excitation is 1, i.e., $m=1$ and $k_1=1$, the joint system-field state $\ket{\Psi(t)}$ in \eqref{Psi_3} is
\begin{equation}\label{oct21-1}
\ket{\Psi(t)}=|\eta_t\rangle\otimes\ket{\Phi_0}+\int_{0}^{t}|\eta_t(t_1)\rangle dB_{\mathrm{in}}^\ast(t_1)\otimes\ket{\Phi_0}.
\end{equation}
 The first term in \eqref{oct21-1} indicates that the single excitation exists in the Tavis-Cummings model. $|\eta_t\rangle$ can be described by
\begin{equation*}\label{oct21-6}
|\eta_t\rangle=\sum_{j=1}^{N}c_j(t)|g_1\cdots e_j\cdots g_N0\rangle+c_{N+2,k}(t)|g_1g_2\cdots g_N1\rangle
\end{equation*}
with the initial conditions $c_k(0)=1$ and  $c_j(0)=c_{N+2,k}(0)=0$ for $1\leq j\leq N, ~ j\neq k$.
%\begin{equation*}\label{oct21-7}\begin{aligned}
%c_k(0)=1, ~ c_j(0)=c_{N+2,k}(0)=0, ~ 1\leq j\leq N, ~ j\neq k.
%\end{aligned}\end{equation*}
On the other hand, the second term in \eqref{oct21-1} means that the Tavis-Cummings model emits a photon into the output field, which can be rewritten as
\begin{equation*}\label{oct21-11}\begin{aligned}
&\int_{0}^{t}|\eta_t(t_1)\rangle dB_{\mathrm{in}}^\ast(t_1)\otimes\ket{\Phi_0} \\
=&c_{N+1}(t)\int_{0}^{t}\varphi(t_1)dB_{\mathrm{in}}^\ast(t_1)|g_1g_2\cdots g_N0\rangle\otimes\ket{\Phi_0}.
\end{aligned}\end{equation*}
Here, $\varphi(t_1)$ is the pulse shape of the single-photon output state. When $t\rightarrow\infty$, $c_{N+1}(\infty)$ denotes the probability amplitude corresponding to the reduced joint state $|g_1g_2\cdots g_N0\rangle\otimes|\Phi_1\rangle$. By the recursive relation \eqref{oct23_+1},  we have $\ket{\eta_t(t_1)} = G(t,t_1) L_1(t_1) G(t,t_1)^{-1}  \ket{\eta_t}$,
%\begin{equation*}\label{eq:jun9_thm_5.1}
%\ket{\eta_t(t_1)} = G(t,t_1) L_1(t_1) G(t,t_1)^{-1}  \ket{\eta_t},
%\end{equation*}
where  $\ket{\eta_t}$ can be calculated via   \eqref{eq:jun9_initial}  and \eqref{oct21-2}. Theorem \ref{state} follows.
\end{remark}
%%
%%
%Remark \ref{Rem6.4} reveals that how the joint system-field superpostion state is generated by the Tavis-Cummings model with one initially excited atom. Moreover, if the initial condition \eqref{oct21-7} is changed to be
%\begin{equation}\label{May30-1}
%c_j(0)=0, ~ c_{N+2,k}(0)=1, ~ 1\leq j\leq N,
%\end{equation}
%which means that the $N$ atoms are initially in the ground state, the cavity contain a single photon and the input filed is in the vacuum state, then the resulting steady-state output field state of the Tavis-Cummings model can be proved to be a single-photon state. In other words, the single photon initially existing in the cavity eventually escapes from the system, which cannot entangle the atoms with the field.

In the following subsections, we apply the above theory to our  Tavis-Cummings model. In this case, $m=1$, $L_i(t)$ in \eqref{+1} is $L=\sqrt{\kappa}a$, and  $H_\mathrm{eff}(t)$ in \eqref{dot_eta} is $H_{\mathrm{TC}} -\ff{\mathrm{i}}{2}L^\ast L$, where $H_{\mathrm{TC}}$ is given in \eqref{H_TC}. For simplicity, we set $\omega_1=\ldots=\omega_N\equiv\omega_r=0$ and $\Gamma_1=\ldots =\Gamma_N=\kappa=1$. Assume that the Tavis-Cummings model contains $R$ initially excited two-level atoms, the cavity is initially empty, and the input is in the vacuum state. In such setting, the number of basis states of the  Tavis-Cummings model is $K= 2^N (R+1)$ as discussed in subsection \ref{subsec:model}. Finally, for better illustration we plot the symmetric pulse shape in  \eqref{eq:n photons}, instead of that in \eqref{eq:aug16_5b}.

%%%%%%%%%%%%%%%%%%%%%%%%%%
%%%%%%%%%%%%%%%%%%%%%%%%%%
\subsection{$N=3$ and $R=1$.} \label{subsec:sim1}

In this case, the total number of basis states is $K=16$. The joint system-field state can be expressed as $\ket{\Psi(t)}=|\eta_t\rangle\otimes\ket{\Phi_0}+\int_{0}^{t}|\eta_t(t_1)\rangle dB_{\mathrm{in}}^\ast(t_1)\ket{\Phi_0}$,
%\begin{equation*}
%\ket{\Psi(t)}=|\eta_t\rangle\otimes\ket{\Phi_0}+\int_{0}^{t}|\eta_t(t_1)\rangle dB_{\mathrm{in}}^\ast(t_1)\ket{\Phi_0},
%\end{equation*}
where by  \eqref{eq:jun9_temp1}, $|\eta_t(t_1,\ldots,t_k)\rangle=\sum_{j=0}^{15}\xi^{j_s}_t(t_1,\ldots,t_k)\ket{j_s}, ~ k=0,1$.
%\begin{equation*}
%|\eta_t(t_1,\ldots,t_k)\rangle=\sum_{j=0}^{15}\xi^{j_s}_t(t_1,\ldots,t_k)\ket{j_s}, ~ k=0,1.
%\end{equation*}
%
%Notice that $\xi_{j_s}^1(t,t_1)$, $j=0,\ldots,15$ contain not only the amplitudes of $\ket{\Psi(t)}$ reduced to $\ket{j_s}\otimes dB_{\mathrm{in}}^\dagger(t_1)\ket{\Phi_0}$ but also the pulse shape of the corresponding output photon state in the steady state with respect to $t_1$. Moreover, we denote the corresponding amplitude of the output photon state in the steady state ($t\rightarrow\infty$) by $\xi_{j_s}^1(\infty,t_1)$, $j=0,\ldots,15$.
%
Let the initial joint system-field state be $\ket{\Psi(0)} = |e_1g_2g_30\rangle\otimes\ket{\Phi_0}$.
%\begin{equation*}\label{sep23a}
%\ket{\Psi(0)}=|\eta_{0}\rangle\otimes\ket{\Phi_0}=|e_1g_2g_30\rangle\otimes\ket{\Phi_0}.
%\end{equation*}
By the recursive relation \eqref{oct23_+1}, we have
%
%Since the effective Hamiltonian $H_{\mathrm{eff}}$ and the coupling operator $L$ are time-invariant, the propagator $V(t)$ and the transition matrix $G(t,\tau)$ can be solved analytically. According to Algorithm \ref{algorithm1}, we have
%\begin{equation}\label{sep23b}\begin{aligned}
%|\dot{\eta}_t^{0}\rangle&=-\mathrm{i}H_{\mathrm{eff}}(t)|\eta_t^{0}\rangle, \\
%|\eta_t^{1}(t_1)\rangle&=G(t,t_1)L(t_1)G(t,t_1)^{-1}|\eta_t^{0}\rangle.
%\end{aligned}\end{equation}
%%In what follows, we set $\omega_r=\omega_s=0$, $g_1=g_2=g_3=1$, $\kappa=1$.
%Solving \eqref{sep23b} with the initial condition
%\begin{equation}
%|\eta^0_{0}\rangle=|e_1g_2g_30\rangle,
%\end{equation}
%yields $\xi_{j_s}^0(t)$, $j=0,1,\ldots,15$ with the non-zero elements
\begin{equation*}\begin{aligned}
&\xi^{1_s}_t=\frac{2}{\sqrt{47}}\left(e^{\frac{-1-\sqrt{47}\mathrm{i}}{4}t}-e^{\frac{-1+\sqrt{47}\mathrm{i}}{4}t}\right), \\
&\xi^{2_s}_t=-\frac{1}{3}+\frac{2\left(\frac{\sqrt{47}+\mathrm{i}}{4}e^{\frac{-1-\sqrt{47}\mathrm{i}}{4}t}+\frac{\sqrt{47}-\mathrm{i}}{4}e^{\frac{-1+\sqrt{47}\mathrm{i}}{4}t}\right)}{3\sqrt{47}}, \\
&\xi^{4_s}_t=\xi^{2_s}_t, \\
&\xi^{8_s}_t=\frac{2}{3}+\frac{2\left(\frac{\sqrt{47}+\mathrm{i}}{4}e^{\frac{-1-\sqrt{47}\mathrm{i}}{4}t}+\frac{\sqrt{47}-\mathrm{i}}{4}e^{\frac{-1+\sqrt{47}\mathrm{i}}{4}t}\right)}{3\sqrt{47}},
\\
&\xi^{0_s}_t(t_1)=-\frac{4\mathrm{i}}{\sqrt{47}}e^{-\frac{1}{4}t_1}\sin\left(\frac{\sqrt{47}}{4}t_1\right), ~ 0\leq t_1\leq t,
\end{aligned}\end{equation*}
and all the others are 0. In the steady state ($t=\infty$), we have $\xi^{2_s}_\infty=\xi^{4_s}_\infty=-\frac{1}{3}, \xi^{8_s}_\infty=\frac{2}{3}$,
%\begin{equation*}
%\xi^{2_s}_\infty=\xi^{4_s}_\infty=-\frac{1}{3}, ~~ \xi^{8_s}_\infty=\frac{2}{3},
%\end{equation*}
which correspond to the basis vectors $|g_1g_2e_30\rangle,  |g_1e_2g_30\rangle,  |e_1g_2g_30\rangle$
%\begin{equation*}
%\Big\{|g_1g_2e_30\rangle, ~~ |g_1e_2g_30\rangle, ~~ |e_1g_2g_30\rangle\Big\},
%\end{equation*}
respectively.
%Moreover, by the recursive relation \eqref{oct23_+1}, we have
%\begin{equation}\label{sep23c}\begin{aligned}
%|\eta_t(t_1)\rangle=\xi^{0_s}_t(t_1)|g_1g_2g_30\rangle,
%\end{aligned}\end{equation}
%where
%\begin{equation}
%\xi^{0_s}_t(t_1)=\xi^{4_s}_{t_2}, ~ 0\leq t_1\leq t,
%\end{equation}
Moreover, as $\lim_{t\to\infty}\int_0^{\infty} |\xi^{0_s}_t(t_1)|^2 dt_1=\frac{1}{3}$,
%\begin{equation*}
%\lim_{t\to\infty}\int_0^{\infty} |\xi^{0_s}_t(t_1)|^2 dt_1=\frac{1}{3},
%\end{equation*}
%%Thus, the normalization condition of $\ket{\Psi(t)}$ in the steady state ($t\rightarrow\infty$) has been verified by
%\begin{equation}
%|\xi_{2_s}^0(\infty)|^2+|\xi_{4_s}^0(\infty)|^2+|\xi_{8_s}^0(\infty)|^2+|\xi_{0_s}^1(\infty,t_1)|^2=1.
%\end{equation}
the steady-state joint system-field state is
\begin{equation}\label{oct8a}\begin{aligned}
&|\Psi(\infty)\rangle=\frac{2}{3}|e_1g_2g_30\rangle\otimes\ket{\Phi_0}-\frac{1}{3}|g_1e_2g_30\rangle\otimes\ket{\Phi_0}
\\
 &\hspace{8ex} -\frac{1}{3}|g_1g_2e_30\rangle\otimes\ket{\Phi_0}+ \frac{1}{\sqrt{3}}|g_1g_2g_30\rangle\otimes |1_{\eta}\rangle,
\end{aligned}\end{equation}
where the pulse shape of the single photon is $\eta(t)=\sqrt{3} \xi^{0_s}_\infty(t)$ up to a global phase.  Finally, the square of the probability amplitude $2/3$ in  \eqref{oct8a} is consistent with the limiting value of the red solid curve in Fig. \ref{two_excited_new}.
%\begin{equation*}
%\eta_{\mathrm{steady}}(t)=\sqrt{3} \xi^{0_s}_t(t_1), \ \ \ 0\leq t_1\leq t.
%\end{equation*}
%Rabi oscillation can be seen in Fig. \ref{oct8_single}  from peaks in the probability distribution $|\eta_{\mathrm{steady}}(t)|^2$ of the steady-state output photon state $\ket{1_{\eta_{\mathrm{steady}}}}$; moreover, it converges to zero around $t=12$. These are consistent with the excitation probability of the first two-level atom plotted in Fig. \ref{single+excited}.

%\begin{figure}[htp!]
%\centering
%\includegraphics[scale=0.5]{oct8_single.png}
%\caption{The comparison between probability distribution  of the steady-state output photon state $\ket{1_{\eta_{\mathrm{steady}}}}$ in \eqref{oct8a} and the excitation probability of the first two-level atom shown by the red solid curve in Fig. \ref{single+excited}.}
%\label{oct8_single}
%\end{figure}

%%%%%%%%%%%%%%%%%%%%%%%%%%
%%%%%%%%%%%%%%%%%%%%%%%%%%
\subsection{$N=2$ and $R=2$.} \label{subsec:sim4}

In this case, the total number of basis states is $K=12$.  The joint system-field state can be expressed as
\begin{equation*}\begin{aligned}
\ket{\Psi(t)}=&|\eta_t\rangle\otimes\ket{\Phi_0}+\int_{0}^{t}|\eta_t(t_1)\rangle dB_{\mathrm{in}}^\ast(t_1)\ket{\Phi_0}  \\
&+\int_{0}^{t}\int_{0}^{t_2}|\eta_t(t_1,t_2)\rangle dB_{\mathrm{in}}^\ast(t_1)dB_{\mathrm{in}}^\ast(t_2)\ket{\Phi_0},
\end{aligned}\end{equation*}
where $|\eta_t(t_1,\ldots,t_k)\rangle=\sum_{j=0}^{11}\xi^{j_s}_t(t_1,\ldots,t_k)\ket{j_s}, ~ k=0,1,2$.
%\begin{equation*}
%|\eta_t(t_1,\ldots,t_k)\rangle=\sum_{j=0}^{11}\xi^{j_s}_t(t_1,\ldots,t_k)\ket{j_s}, ~ k=0,1,2.
%\end{equation*}
Let the initial joint system-field state be $\ket{\Psi(0)} = |e_1e_20\rangle\otimes\ket{\Phi_0}$.
%\begin{equation*}\label{sep30a}
%\ket{\Psi(0)}=|\eta_{0}\rangle\otimes\ket{\Phi_0}=|e_1e_20\rangle\otimes\ket{\Phi_0}.
%\end{equation*}
%According to \eqref{dot_eta} and \eqref{oct23_+1}, we have
%\begin{equation}\label{sep30b}\begin{aligned}
%|\dot{\eta}_t^{0}\rangle&=-\mathrm{i}H_{\mathrm{eff}}(t)|\eta_t^{0}\rangle, \\
%|\eta_t^{1}(t_1)\rangle&=G(t,t_1)L(t_1)G(t,t_1)^{-1}|\eta_t^{0}\rangle, \\
%|\eta_t^{2}(t_1,t_2)\rangle&=G(t,t_2)L(t_2)G(t,t_2)^{-1}|\eta_t^{1}(t_1)\rangle.
%\end{aligned}\end{equation}
%Similarly, we set $\omega_r=\omega_s=0$, $g_1=g_2=1$, $\kappa=1$.
By Theorem \ref{thm:ss_photon},  the steady-state joint system-field state ($t\rightarrow\infty$) is
\begin{equation}\label{oct8c}
|\Psi(\infty)\rangle=|g_1g_20\rangle\otimes|2_{\eta}\rangle,
\end{equation}
where $|2_{\eta}\rangle$ is the 2-photon output state. By the recursive relation \eqref{oct23_+1}, the pulse shape of the 2-photon state $|2_{\eta}\rangle$ is
\begin{equation*}\begin{aligned}
\eta(t_1,t_2) =& \xi^{0_s}_\infty (t_1,t_2)\\
=&(0.0228787-0.00785441\mathrm{i})
e^{\mu_1t_1-\mu_4t_2} \\
+&(0.319599+0.0291387\mathrm{i})
e^{\mu_2 t_1-\mu_4t_2} \\
-&(0.342477+0.0212843\mathrm{i})
e^{\mu_3t_1-\mu_4t_2}+\mathrm{c.c.},
\end{aligned}\end{equation*}
where ``$\mathrm{c.c.}$'' means complex conjugate, and $\mu_1=-0.336506+3.79453\mathrm{i}$, $\mu_2=-0.336506-1.01065\mathrm{i}$,  $\mu_3=-0.076987+1.39194\mathrm{i}$, and $\mu_4=0.25+1.39194\mathrm{i}$.
%\begin{equation*}\label{oct9f}\begin{aligned}
%&\mu_1=-0.336506+3.79453\mathrm{i},~~\mu_2=-0.336506-1.01065\mathrm{i}, \\
%&\mu_3=-0.076987+1.39194\mathrm{i},~~\mu_4=0.25+1.39194\mathrm{i}.
%\end{aligned}\end{equation*}
It can be easily verified that $
\int_{0}^{\infty}\int_{0}^{t_2}|\eta(t_1,t_2)|^2dt_1dt_2=1$.
The probability distribution of the steady-state output two-photon state $|2_{\eta}\rangle$ in   \eqref{oct8c} is simulated in Fig. \ref{oct8_2photons2}.

\begin{figure}[htp!]
\centering
\includegraphics[scale=0.6]{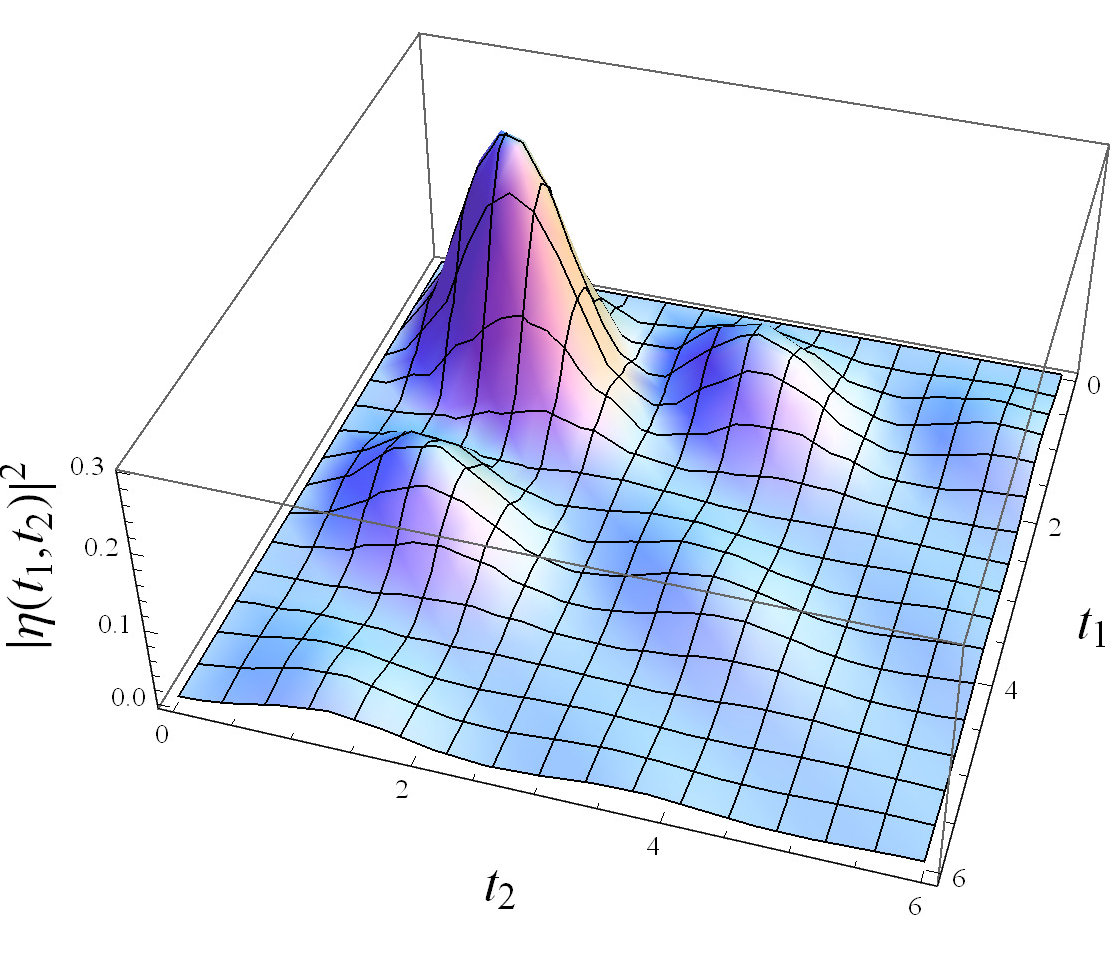}
\caption{Probability distribution of  $|2_{\eta}\rangle$ in  \eqref{oct8c}.}
\label{oct8_2photons2}
\end{figure}

%%%%%%%%%%%%%%%%%%%%%%%%%%
%%%%%%%%%%%%%%%%%%%%%%%%%%
\subsection{$N=3$ and $R=2$.} \label{subsec:sim2}

In this case, the total number of basis states is $K=24$.  The joint system-field state can be expressed as
\begin{equation*}\begin{aligned}
\ket{\Psi(t)}=&|\eta_t\rangle\otimes\ket{\Phi_0}+\int_{0}^{t}|\eta_t(t_1)\rangle dB_{\mathrm{in}}^\ast(t_1)\ket{\Phi_0} \\
&+\int_{0}^{t}\int_{0}^{t_2}|\eta_t(t_1,t_2)\rangle dB_{\mathrm{in}}^\ast(t_1)dB_{\mathrm{in}}^\ast(t_2)\ket{\Phi_0},
\end{aligned}\end{equation*}
where $|\eta_t(t_1,\ldots,t_k)\rangle=\sum_{j=0}^{23}\xi^{j_s}_t(t_1,\ldots,t_k)\ket{j_s}, ~ k=0,1,2$.
%\begin{equation*}
%|\eta_t(t_1,\ldots,t_k)\rangle=\sum_{j=0}^{23}\xi^{j_s}_t(t_1,\ldots,t_k)\ket{j_s}, ~ k=0,1,2.
%\end{equation*}
Let  the initial joint system-field state be $\ket{\Psi(0)} = |g_1e_2e_30\rangle\otimes\ket{\Phi_0}$.
%\begin{equation*}\label{sep24a}
%\ket{\Psi(0)}=|\eta_{0}\rangle\otimes\ket{\Phi_0}=|g_1e_2e_30\rangle\otimes\ket{\Phi_0}.
%\end{equation*}
%Similarly, we set $\omega_r=\omega_s=0$, $g_1=g_2=g_3=1$, $\kappa=1$.
According to the recursive relation \eqref{oct23_+1}, we have the steady-state joint system-field state
\begin{equation}\label{oct8b}\begin{aligned}
|\Psi(\infty)\rangle
=&\frac{2}{3}|e_1g_2g_30\rangle\otimes|1_{\eta^{{12}_s}}\rangle +\frac{1}{3}|g_1e_2g_30\rangle\otimes|1_{\eta^{{6}_s}}\rangle \\
&+\frac{1}{3}|g_1g_2e_30\rangle\otimes|1_{\eta^{{3}_s}}\rangle +\frac{\sqrt{3}}{3}|g_1g_2g_30\rangle\otimes|2_{\eta^{{0}_s}}\rangle,
\end{aligned}\end{equation}
where $\eta^{{12}_s}(t)= \ff{3}{2}\xi^{12_s}_\infty(t)= 0.5163975\left(e^{\lambda^\prime t}-e^{\lambda^{\prime \ast}t}\right)$ with $\lambda^\prime=-0.25+0.968246\mathrm{i}$, $\eta^{{6}_s}(t)=\eta^{{3}_s}(t)=-\eta^{{12}_s}(t)$,
%\begin{equation*}\label{oct9a}\begin{aligned}
%\eta^{{12}_s}(t)&=\frac{0.344265}{\sqrt{0.444444}}\left(e^{\lambda_1^1t}-e^{\lambda_1^{1\ast}t}\right),
%\\
%\eta^{{6}_s}(t)&=\eta^{{3}_s}(t)=-\eta^{{12}_s}(t)
%\end{aligned}\end{equation*}
%\begin{equation*}\label{oct9d}\begin{aligned}
%\lambda_1^1=-0.25+0.968246\mathrm{i},
%\end{aligned}\end{equation*}
and
\vspace{-1ex}
\begin{equation*}\begin{aligned}
\eta^{{0}_s}(t_1,t_2) =& \sqrt{3}\xi^{{0}_s}_\infty(t_1,t_2) \\
=&(0.338571+0.0140041\mathrm{i})e^{\lambda_1^\prime  t_1-\lambda_4^\prime t_2} \\
&+(0.0155773-0.00341903\mathrm{i})e^{\lambda_2^\prime t_1-\lambda_4^\prime t_2} \\
&-(0.354149+0.0105851\mathrm{i})e^{\lambda_3^\prime t_1-\lambda_4^\prime t_2}+\mathrm{c.c.}
\end{aligned}
\vspace{-.8ex}
\end{equation*}
with $\lambda_1^\prime =-0.301227-1.40985\mathrm{i}$, $\lambda_2^\prime =-0.301227+4.83767\mathrm{i}$,
$\lambda_3^\prime =-0.147546+1.71391\mathrm{i}$, and $\lambda_4^\prime =0.25+1.71391\mathrm{i}$.
%\begin{equation*}\label{oct9e}\begin{aligned}
%&\lambda_1=-0.301227+1.40985\mathrm{i},~~\lambda_2=-0.301227+4.83767\mathrm{i}, \\
%&\lambda_3=-0.147546+1.71391\mathrm{i},~~\lambda_4=0.25+1.71391\mathrm{i}.
%\end{aligned}\end{equation*}
%
%The pulse shape of the steady-state single-photon output state $|1_\eta\rangle$ corresponding to $\xi_{{6}_s}^1(\infty,t_1)|g_1e_2g_30\rangle\otimes|1_\eta\rangle$ is
%\begin{equation}\label{oct9b}\begin{aligned}
%\eta_{{6}_s}^1(t)=\frac{-0.172133}{\sqrt{0.111111}}\left(e^{\lambda_1^1t}-e^{\lambda_1^{1\ast}t}\right).
%\end{aligned}\end{equation}
%
%The pulse shape of the steady-state single-photon output state $|1_\eta\rangle$ corresponding to $\xi_{{3}_s}^1(\infty,t_1)|g_1g_2e_30\rangle\otimes|1_\eta\rangle$ is
%\begin{equation}\label{oct9c}\begin{aligned}
%\eta_{{3}_s}^1(t)=\frac{-0.172133}{\sqrt{0.111111}}\left(e^{\lambda_1^1t}-e^{\lambda_1^{1\ast}t}\right)=\eta_{{6}_s}^1(t).
%\end{aligned}\end{equation}
%
The square of the probability amplitude $\frac{2}{3}$ in \eqref{oct8b} is consistent with the limiting value of the red dashed curve in Fig. \ref{two_excited_new}. It is clear in Fig. \ref{oct8_3singles} that  $\ket{1_{\eta}}$ has more oscillations than  $\ket{1_{\eta^{{12}_s}}}$.
%Finally, the probability distribution of the steady-state output two-photon state component $|2_{\eta^{{0}_s}}\rangle$ is shown in Fig. \ref{oct8_2photons1}.

%\vspace{-1ex}
\begin{figure}[htp!]
\centering
\includegraphics[scale=0.4]{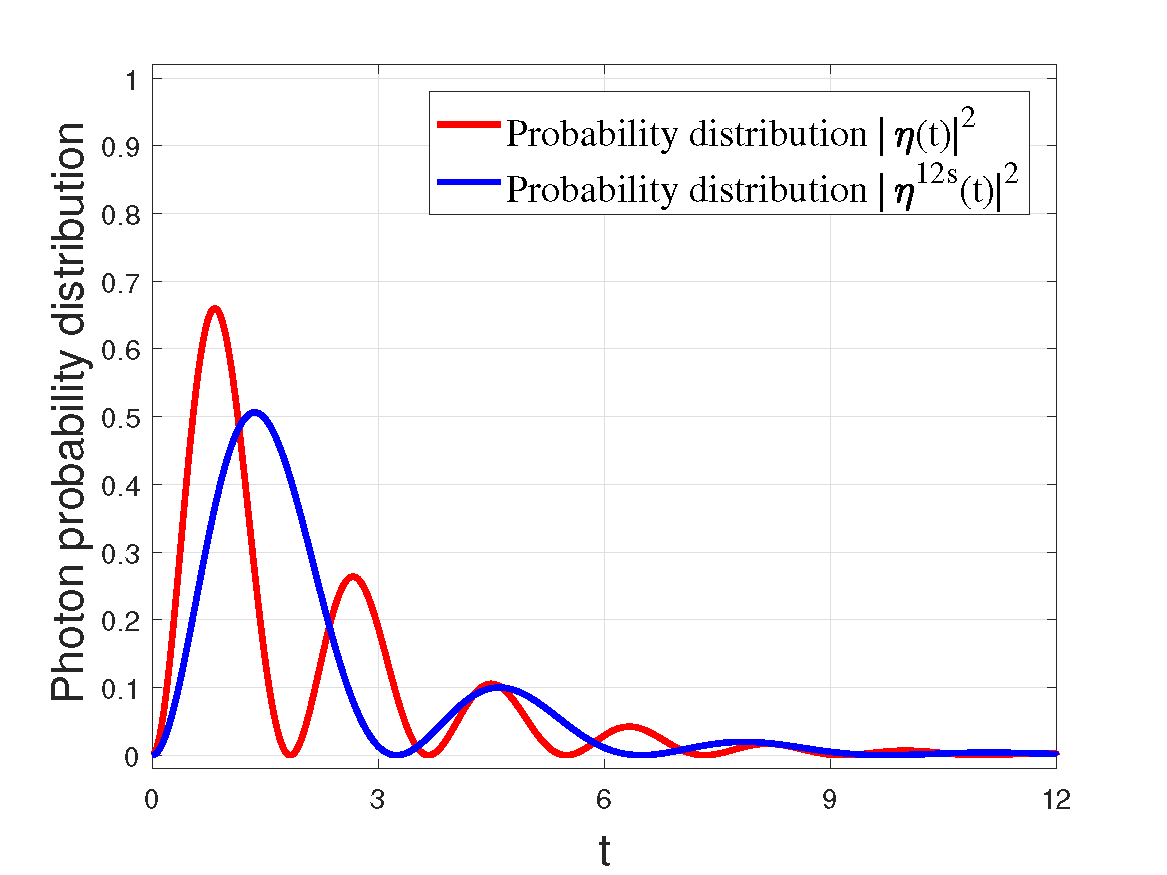}
\caption{The comparison between probability distributions of the steady-state single-photon  output states $\ket{1_{\eta}}$ in  \eqref{oct8a}  and $\ket{1_{\eta^{{12}_s}}}$ in  \eqref{oct8b}.}
\label{oct8_3singles}
\end{figure}

%\vspace{-1ex}
\begin{figure}[htp!]
\centering
\includegraphics[scale=0.6]{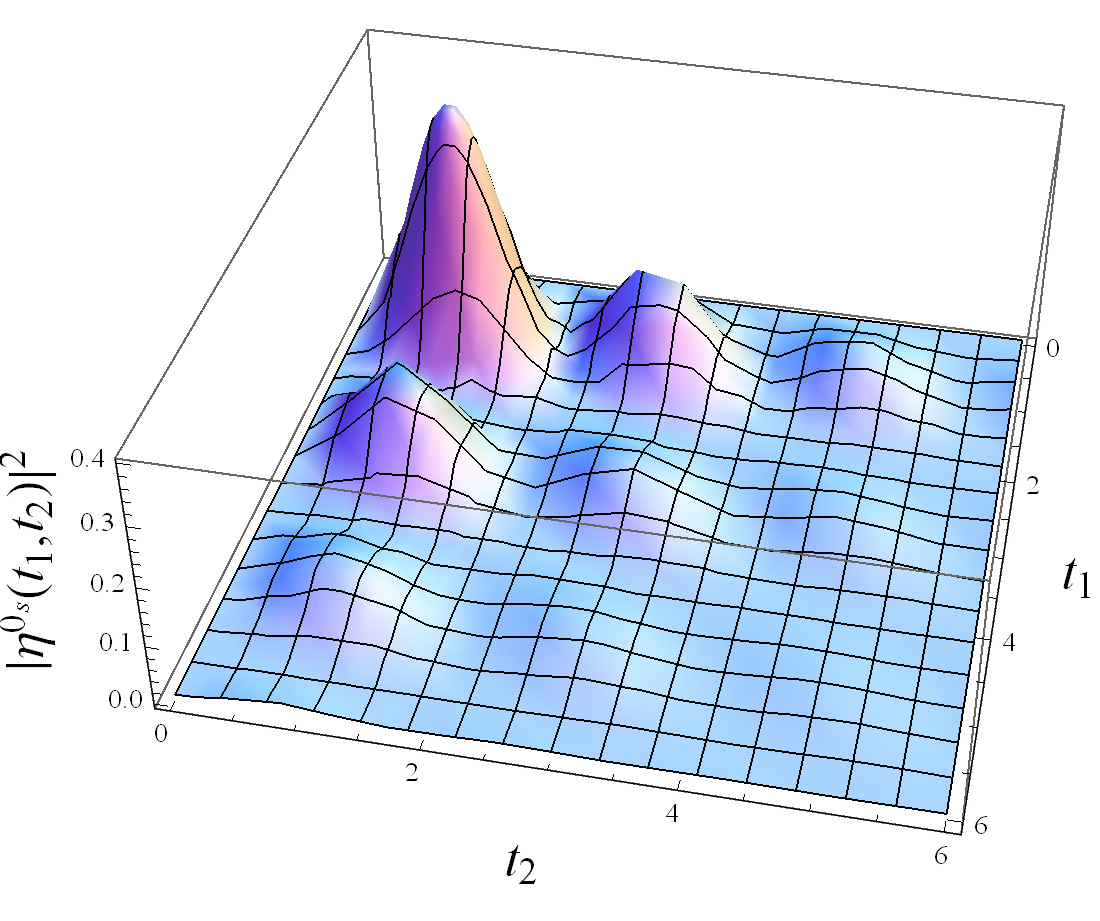}
\caption{Probability distribution of the two-photon  output state $|2_{\eta^{{0}_s}}\rangle$ in  \eqref{oct8b}.}
\label{oct8_2photons1}
\end{figure}

Both Figs. \ref{oct8_2photons2} and \ref{oct8_2photons1} show the probability distributions of a two-photon output state. The one generated by the Tavis-Cummings system with $3$ atoms (shown in Fig. \ref{oct8_2photons1}) oscillates more rapidly than that with $2$ atoms (see Fig. \ref{oct8_2photons2}). This result indicates that the Rabi oscillation is enhanced by adding more atoms, which is consistent with the discussions in Fig. \ref{distri1} and the observations in \cite[Fig. 5]{CS14}.

%%%%%%%%%%%%%%%%%%%%%%%%%%
%%%%%%%%%%%%%%%%%%%%%%%%%%
\subsection{$N=3$ and $R=3$.} \label{subsec:sim3}

In this case, the total number of basis states is $K=32$ and the joint system-field state can be expressed as
\begin{equation*}\begin{aligned}
&\ket{\Psi(t)}=|\eta_t\rangle\otimes\ket{\Phi_0}+\int_{0}^{t}|\eta_t(t_1)\rangle dB_{\mathrm{in}}^\ast(t_1)\ket{\Phi_0} \\
+&\int_{0}^{t}\int_{0}^{t_2}|\eta_t(t_1,t_2)\rangle dB_{\mathrm{in}}^\ast(t_1)dB_{\mathrm{in}}^\ast(t_2)\ket{\Phi_0} \\
+&\int_{0}^{t}\int_{0}^{t_3}\int_{0}^{t_2}|\eta_t(t_1,t_2,t_3)\rangle dB_{\mathrm{in}}^\ast(t_1)dB_{\mathrm{in}}^\ast(t_2)dB_{\mathrm{in}}^\ast(t_3)\ket{\Phi_0},
\end{aligned}\end{equation*}
where $|\eta_t(t_1,\ldots,t_k)\rangle=\sum_{j=0}^{31}\xi^{j_s}_t(t_1,\ldots,t_k)\ket{j_s}, ~ k=0,1,2,3$.
%\begin{equation*}
%|\eta_t(t_1,\ldots,t_k)\rangle=\sum_{j=0}^{31}\xi^{j_s}_t(t_1,\ldots,t_k)\ket{j_s}, ~ k=0,1,2,3.
%\end{equation*}
Since the three atoms are all initially excited, the initial joint system-field state is $\ket{\Psi(0)}=|e_1e_2e_30\rangle\otimes\ket{\Phi_0}$.
%\begin{equation*}\label{sep26c}
%\ket{\Psi(0)}=|\eta_{0}\rangle\otimes\ket{\Phi_0}=|e_1e_2e_30\rangle\otimes\ket{\Phi_0}.
%\end{equation*}
%%According to \eqref{dot_eta} and \eqref{oct23_+1}, we have
%\begin{equation}\label{sep26d}\begin{aligned}
%|\dot{\eta}_t^{0}\rangle&=-\mathrm{i}H_{\mathrm{eff}}(t)|\eta_t^{0}\rangle, \\
%|\eta_t^{1}(t_1)\rangle&=G(t,t_1)L(t_1)G(t,t_1)^{-1}|\eta_t^{0}\rangle, \\
%|\eta_t^{2}(t_1,t_2)\rangle&=G(t,t_2)L(t_2)G(t,t_2)^{-1}|\eta_t^{1}(t_1)\rangle, \\
%|\eta_t^{3}(t_1,t_2,t_3)\rangle&=G(t,t_3)L(t_3)G(t,t_3)^{-1}|\eta_t^{2}(t_1,t_2)\rangle. \\
%\end{aligned}\end{equation}
%Similarly, we set $\omega_r=\omega_s=0$, $g_1=g_2=g_3=1$, $\kappa=1$.
By Theorem \ref{thm:ss_photon},  the steady-state joint system-field state  is
\begin{equation}\label{oct10a}
|\Psi(\infty)\rangle=|g_1g_2g_30\rangle\otimes|3_\eta\rangle,
\end{equation}
where $|3_\eta\rangle$ is the three-photon output state, whose probability distribution  is shown  in Fig. \ref{oct8_3photons}.

\begin{figure}[htp!]
\centering
\includegraphics[scale=0.3]{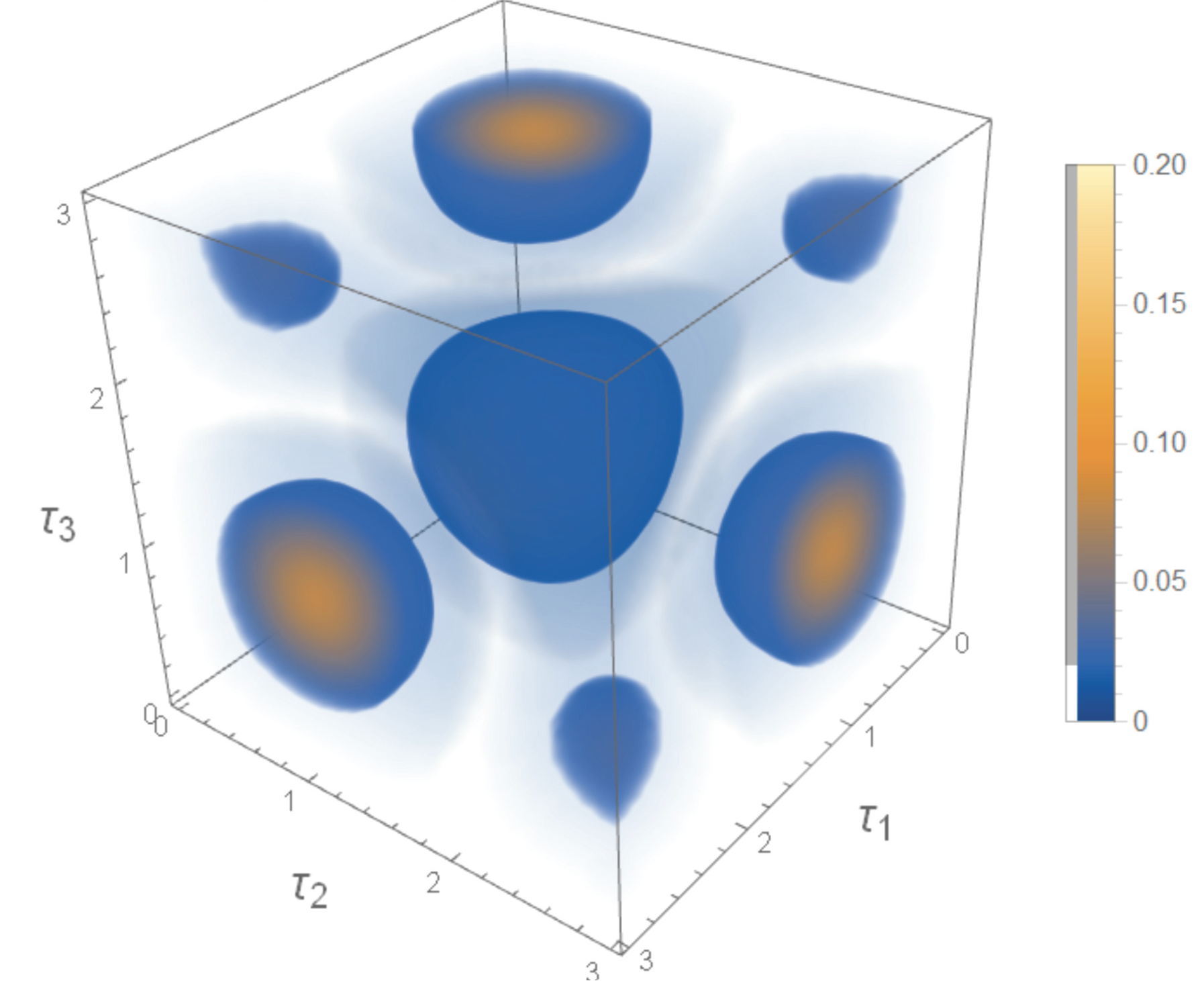}
\caption{Probability distribution of the three-photon  output state $|3_\eta\rangle$ in \eqref{oct10a}.}
\label{oct8_3photons}
\end{figure}

\section{Conclusion}\label{sec:conclusion}
In this paper, we have studied the Tavis-Cummings model. Specifically, we have applied quantum linear systems theory to reveal typical features of the Tavis-Cummings model, the analytical expression has been derived for the output single-photon state of a Tavis-Cummings system in response to a single-photon input. We have also proposed a computational framework  to derive the analytic form of the superposition state of  the system and field.
%
%By Theorem \ref{theoutput}, the exact expression of the single-photon pulse shapes $\eta$ in Fig. \ref{distri1} can be obtained. However, the non-resonant case given by the blue dotted curve is quite different from the resonant case.  We do not yet have an intuitive physical interpretation to it.
%
%
%As shown by the red dashed curve in Fig. \ref{single+excited}, the input photon tends to decrease the excitation probability of the first two-level atom when it is initialized in the excited state. This is different from the $N=1$ case (the Jaynes-Cummings system). The red dashed curve is obtained by numerically solving the single-photon master equation \eqref{master}, but we do not have a simple physical interpretation of this simulation result.
%
%
%The red curves in Fig. \ref{two_excited_new} have the same final steady-state value, this has been theoretically  confirmed in subsection \ref{subsec:sim2}. But  we do not have a simple physical interpretation of this simulation result.
%
Note that the linear quantum systems' approach  is only applicable when there is only one excitation. Future studies are to be done on subradiance and superradiance in the multi-excitation case.

%%%%%%%%%%%%%%%%%%%%%%%%%%
%%%%%%%%%%%%%%%%%%%%%%%%%%
%%%%%%%%%%%%%%%%%%%%%%%%%%
\section*{Appendix}\label{sec:appendix}
As concepts and properties of quantum linear passive systems are used in this paper, in particular Section \ref{sec:single_excitation}, in this appendix we collect some of them for readability of this paper.

For the quantum linear system \eqref{linear_system}, from \eqref{eq:prl-111} we know  the  transfer function $G[s]$ is an {\it all-pass} function \cite[pp. 357]{ZDG96}. Actually, as shown in Remark \ref{rem:passivity}, system \eqref{linear_system} is passive. If we take expectation on both sides of  \eqref{linear_system} with respect to the initial joint system-field state
(which is a unit vector in a Hilbert space), we get a {\it classical} linear system
%
%
%%General speaking, if a linear system does not require an external source of energy for its operation, then it is called a {\it passive} linear system. As given in Section \ref{subsec:linear}, the system Hamiltonian of  the quantum linear system \eqref{linear_system} is $H=\omega_r a^\ast a + \sum_{j=1}^{N}\left[\omega_ja_j^\ast a_j+g_j(a^\ast a_j+a_j^\ast a)\right]$, and its coupling to the external field is $L=\sqrt{\kappa}a$. Thus,  there is no energy input to this system, and hence system \eqref{linear_system} is a quantum linear passive system.
%
%
%Given an operator $X$ on the system space $\mathfrak{h}$, denote by $\left\langle X(t)\right\rangle $ the
%expected value of $X(t)$ with respect to the initial system-field joint state
%(which is a unit vector in the Hilbert space $\mathfrak{h}\otimes \mathfrak{F}$).
%Then system \eqref{linear_system} gives rise to
%the following {\it classical} linear system
%\begin{equation}\label{linear_system_mean}\begin{aligned}
%d\left\langle\bar{a}(t)\right\rangle&=A\left\langle\bar{a}(t)\right\rangle dt+Bd\left\langle B_{\mathrm{in}}(t)\right\rangle, \\
%d\left\langle B_{\mathrm{out}}(t)\right\rangle&=C\left\langle\bar{a}(t)\right\rangle dt+ d\left\langle B_{\mathrm{in}}(t)\right\rangle.
%\end{aligned}\end{equation}
\begin{equation}\label{linear_system_mean}\begin{aligned}
\left\langle\dot{\bar{a}}(t)\right\rangle&=A\left\langle\bar{a}(t)\right\rangle +B\left\langle b_{\mathrm{in}}(t)\right\rangle, \\
\left\langle b_{\mathrm{out}}(t)\right\rangle&=C\left\langle\bar{a}(t)\right\rangle dt+ \left\langle b_{\mathrm{in}}(t)\right\rangle.
\end{aligned}\end{equation}
Thus we can define controllability, observability, and Hurwitz stability for the quantum linear system \eqref{linear_system} using those for the classical linear system \eqref{linear_system_mean}.

\begin{definition}\label{def:stab_ctrb_obsv}
The quantum linear passive system \eqref{linear_system} is said to be \textit{Hurwitz stable} (resp. \textit{controllable}, \textit{observable}) if the corresponding classical linear system  \eqref{linear_system_mean} is \textit{Hurwitz stable} (resp. \textit{controllable}, \textit{observable}).
\end{definition}

The following result reveals the structure of quantum linear passive systems; see \cite{GZ15,ZGPG18} for more details.

\begin{proposition} \label{prop:3nov}
The quantum linear passive system \eqref{linear_system} has the following properties.
\begin{description}
\item[(i)] Its Hurwitz stability, controllability and observability are equivalent to each other.
\item[(ii)] It has only controllable and observable ($co$) subsystem and uncontrollable and unobservable ($\bar{c}\bar{o}$) subsystem. Moreover, each $\bar{c}\bar{o}$ subsystem is a closed (namely isolated) quantum system; see Eq. \eqref{system_linear} and  Corollary \ref{cor:kalman}.
\item[(iii)] Its poles corresponding to a $\bar{c}\bar{o}$ subsystem are all on the imaginary axis, while its poles corresponding to a $co$ subsystem are on the open left half of the complex plane.
\end{description}
\end{proposition}

%According to Proposition \ref{prop:minimal_C_Omega},  a  $\bar{c}\bar{o}$ subsystem of a quantum linear passive system is an isolated subsystem, hence it will not affect systems' input-output behavior. Moreover, a $co$ subsystem is Hurwitz stable. As a result of these two properties, the following result is a direct consequence of \cite[Theorem 5]{ZJ13} which was derived for Hurwitz stable quantum linear systems.
%
%
%\begin{proposition} \label{prop:minimal_C_Omega}
%If the quantum linear passive system \eqref{linear_system} is initialized in the vacuum state (namely, no photons in the system), and is driven by a single-photon input state $|1_{\xi}\rangle$, then its steady-state output field state is a single-photon state with the frequency-domain pulse shape
%\begin{equation*}\label{outputspe_1nov}
%\eta[\mathrm{i}\omega]=G[\mathrm{i}\omega]\xi[\mathrm{i}\omega],
%\end{equation*}
%where the transfer function $G[s]$ is given by \eqref{transfer}.
%\end{proposition}

%\vspace{-1ex}
%%%%%%%%%%%%%%%%%%%%%%%%%%
\bibliographystyle{IEEEtran}

\end{document}